\documentclass[11pt]{article}

\newcommand{\ignore}[1]{}%

\usepackage[ruled,vlined,resetcount,linesnumbered]{algorithm2e}
\SetKw{Break}{break}
\usepackage{amsmath}
\usepackage{amssymb}
\usepackage{amsthm}
\usepackage{tabularx}
\usepackage{environ}
\usepackage{graphicx}
\usepackage{complexity}
\usepackage{dsfont}
\usepackage[margin=1in]{geometry}
\usepackage[dvipsnames]{xcolor}
\usepackage{url}

\usepackage{footnote}
\makesavenoteenv{tabular}
\makesavenoteenv{table}

\newcommand{\abs}[1]{|#1| }

\let\emptyset\varnothing

\theoremstyle{plain}
\newtheorem{theorem}{Theorem}
\newtheorem*{theorem*}{Theorem}
\newtheorem{lemma}[theorem]{Lemma}
\newtheorem{claim}[theorem]{Claim}

\newtheorem{corollary}[theorem]{Corollary}

\theoremstyle{definition}

\newtheorem{hypothesis}{Hypothesis}
\newtheorem{oq}{Open Question}

\theoremstyle{remark}
\newtheorem*{remark}{Remark}

\newcommand{\ED}{\operatorname{ED}}
\newcommand{\LCS}{\operatorname{LCS}}
\newcommand{\acost}{\operatorname{cost}}
\newcommand{\AG}{\operatorname{GA}}
\newcommand{\NVG}{\operatorname{NVG}}
\newcommand{\dist}{\operatorname{dist}}

\newcommand{\pretime}{t}

\newcommand{\eps}{\epsilon}

\newcommand{\cA}{\mathcal{A}}
\newcommand{\cB}{\mathcal{B}}
\newcommand{\cC}{\mathcal{C}}
\newcommand{\cE}{\mathcal{E}}
\newcommand{\cH}{\mathcal{H}}

\providecommand{\cL}{\mathcal{L}}
\newcommand{\cN}{\mathcal{N}}
\newcommand{\cW}{\mathcal{W}}

\newcommand{\cAs}{{\cA_{\textsc{Sparse}}}}
\newcommand{\cAb}{{\cA_{\textsc{Bad}}}}
\newcommand{\cAc}{{\cA_{\textsc{Covered}}}}
\newcommand{\cBt}{{\cB_{\tau}}}

\newcommand{\set}[1]{\{#1\}}
\newcommand{\anext}{a.\text{next}}
\newcommand{\aprev}{a.\text{prev}}
\newcommand{\iprev}{i.\text{prev}}
\newcommand{\jprev}{j.\text{prev}}

\renewcommand{\R}{\mathbb{R}}

\newcommand{\ttau}{t_\tau}

\title{Does Preprocessing help in Fast Sequence Comparisons?}
\author{
  Elazar Goldenberg%
  \thanks{The Academic College of Tel Aviv-Yaffo. 
    Email: \texttt{elazargo@mta.ac.il}
  }
  \and
  Aviad Rubinstein%
  \thanks{Stanford University.
    Email: \texttt{aviad@cs.stanford.edu}
  }
  \and
  Barna Saha%
  \thanks{University of California Berkeley. Work partially supported by an 
  NSF CAREER Award 1652303, NSF HDR TRIPODS Grant 1934846 and an Alfred P. Sloan Fellowship.
    Email: \texttt{barnas@berkeley.edu}
  }
}

\date{}

\begin{document}

\maketitle
\thispagestyle{empty}
\begin{abstract}

We study edit distance computation with preprocessing: the preprocessing algorithm acts on each string separately, and then the query algorithm takes as input the two preprocessed strings.
This model is inspired by scenarios where we would like to compute edit distance between many pairs in the same pool of strings. 

Our results include:
\begin{description}
\item[Permutation-LCS] If the LCS between two permutations has length $n-k$, we can compute it {\em exactly} with $O(n \log(n))$ preprocessing and $O(k \log(n))$ query time.
\item[Small edit distance] For general strings, if their edit distance is at most $k$, we can compute it {\em exactly} with $O(n\log(n))$ preprocessing and $O(k^2 \log(n))$ query time. 
\item[Approximate edit distance] For the most general input, we can approximate the edit distance to within factor $(7+o(1))$ with preprocessing time $\tilde{O}(n^2)$ and query time $\tilde{O}(n^{1.5+o(1)})$. 
\end{description}

All of these results significantly improve over the state of the art in edit distance computation without preprocessing. Interestingly, by combining ideas from our algorithms with preprocessing, we provide new improved results for approximating edit distance without preprocessing in subquadratic time.

\end{abstract}

\pagebreak

\section{Introduction}

Edit distance (aka Levenshtein distance)~\cite{Lev65} and longest common subsequence are widely used 
distance measures between pairs of strings, over some alphabet $\Sigma$. They find applications in several 
fields like computational biology, pattern recognition, text processing, information retrieval and many more. 
The edit distance between $A$ and $B$, denoted by $\ED(A,B)$, is defined as the minimum number of character 
insertions, deletions, and substitutions needed for converting $A$ into $B$. The longest common subsequence
of $A$ and $B$, denoted by $\LCS(A,B)$, is defined as the longest subsequence common to $A$ and $B$. 
A simple dynamic program solves this problem in quadratic time. 
Moreover under reasonable hardness assumptions like SETH and BP-SETH no real subquadratic time algorithm for 
these problems exists~\cite{BK15,BI18,AHWW16,AB18}.

%
While dealing with huge amounts of data (such as DNA chains, enormous storage,
etc.), quadratic-time algorithms are unaffordable.
This raised an active and extensive line of work on moving from quadratic-time exact computation towards (near)-linear time for approximation algorithms~\cite{BEKMRRS03,BJKK04,BES06,AO12,AKO10,BEGHS18,CDGKS18,CGKK18,Andoni19,BR19,KS19,HSSS19,HRS19,RSSS19,RS20}, and even designing {\em sub}-linear time algorithms for special cases such as restriction on the distance between the input sequences~\cite{BEKMRRS03, AO12,GKS19} or permutations~\cite{CK06,AN10,SS17, NRRS19,BCLW19,RSSS19}.

In many of these applications, a large number of very long strings from a database must be compared among each other (such as comparative genomics, comparing text corpora for documents similarity etc.). For example, in {\em string similarity join}, which is a fundamental problem in databases, one needs to find all pairs of strings (e.g., genome sequences) in a database that are close with respect to edit distance \cite{BZ16}.This in particular motivates developing sub-linear time algorithms. 
But, unfortunately even under strong assumptions, the known guarantees for sub-linear time algorithms (including recent works by the authors) are unsatisfactory. For example, recent work~\cite{GKS19} requires $\Theta(\frac{n}{k}+k^3)$-time, and with a highly non-trivial algorithm can barely distinguish between edit distance $k$ and $k^2$. Even when the strings are both permutations and $k$-close to each other, \cite{AN10}'s nearly-optimal algorithm runs in time $\tilde{O}(\frac{n}{k} + \sqrt{n})$ and still only approximates the edit distance (to within some large constant factor).
In part this is due to strong lower bounds: for example, when the edit distance is $k \ll n$, in order to have any chance of observing any difference between the strings, the algorithm must see $\Omega(\frac{n}{k})$ characters. 

{\bf Our main contribution is a simple and natural augmentation to the standard model: {\em preprocessing}.} Formally, we consider two parties that preprocess each input string independently, and then in a query phase they jointly (approximately) compute an optimal alignment.
Because the preprocessing of the two strings is done independently, 
(i) the same preprocessing of one string can be useful for many comparisons, and (ii) the preprocessing step can be fully parallelized in any distributed system.

In this paper we raise the question of whether preprocessing the input can accelerate the computation 
of the edit distance between input strings and computing their longest common subsequence. 
We affirmatively answer these questions by providing several algorithms that beat the state of 
the art algorithms where no preprocessing is allowed. Our results include faster algorithms for the tasks of exact computation 
of edit distance and permutation LCS. We also provide a better trade off between running time and approximation factor
for edit distance approximation.

We note in particular that when the preprocessing runs in near-linear time (as is the case with all our sublinear-time algorithms), it is essentially for free in the sense that it is barely more than it took to record and store the inputs in the first place. Even when preprocessing takes super-linear time, it could be much more cost-effective to have it when dealing with large number of strings. Preprocessing captures a middle ground between (i) aforementioned works on (approximate) edit distance between two long strings; and (ii) works on approximate closest pair or nearest neighbor among a large number of short strings~\cite{ADGIR03,Indyk04,CK06,OR07,LDHO14,ARW17,Rubinstein18,CGLRR19}. Our preprocessing algorithms are most appealing when both the length and number of strings are large. 

Preprocessing is also closely related to sketching \cite{BZ16,BJKK04}. With an efficient sketching algorithm, we can preprocess a string to compute a small-sized sketch and then only compare the sketches during querying. The state of the art result in edit distance sketching has a preprocessing time of $\tilde{O}(nk^2)$ and query time of $poly(k\log{n})$ \cite{BZ16}. Our algorithms get significantly better trade-offs. There are numerous works on related but different models such as computing embedding of edit distance \cite{ADGIR03,OR07,CK06,CGK16}, document exchange protocols \cite{J12,BZ16,H19} and error-correcting codes for insertions and deletions \cite{BGZ16,HRS19,H19}.



\subsection*{Contributions.}
In the preprocessing model we provide much faster and simpler
algorithms that output much better alignments:
\begin{description}
\item[Permutation-LCS] If the LCS between two permutations has length $n-k$, we can compute it {\em exactly} with $O(n \log(n))$ preprocessing and $O(k \log(n))$ query time. Contrast this result with~\cite{AN10} where in $\tilde{O}(\frac{n}{k} + \sqrt{n})$, the ulam distance can be approximated to within a large constant factor.
\item[Small edit distance] For general strings, if their edit distance is at most $k$, we can compute it {\em exactly} with $O(n\log(n))$ preprocessing and $O(k^2 \log(n))$ query time. Contrast this result with~\cite{GKS19} where in $\tilde{O}(\frac{n}{k} + k^3)$ time, one can distinguish if edit distance is below $k$ or above $\Theta(k^2)$.
\item[Approximate edit distance] For the most general input, we can approximate the edit distance to within factor $(7+o(1))$ with preprocessing time $O(n^2 \log(n))$ and query time $O(n^{1.5+o(1)})$. Contrast this result with~\cite{Andoni19} where a $f(\eps)$-approximation for edit distance can be computed in time $O(n^{1.5+\eps})$ ($f(\eps)$ goes to infinity as $\eps$ decreases).
\item[What if we only preprocess one string?] This setting is much harder, but we can still beat state of the art without preprocessing, namely distinguish $k$ vs $3k^2$ with $\tilde{O}(n)$ preprocessing and $\tilde{O}(n/k + k^2)$ query time.
\end{description}

These strong improvements run contrary to the fine-grained complexity rule of thumb that preprocessing inputs does not help~\cite{WW18}. We also formalize a few conditional hardness results establishing limitations of preprocessing for fast string alignment:
\begin{description}
\item[Exact alignment] We show that assuming (BP)-SETH, even after arbitrary polynomial-time preprocessing, computing edit distance or LCS exactly requires near-quadratic query time. 
\item[Approximate edit-distance] We show that if we can $\alpha$-approximate edit distance in truly-subquadratic query time  with arbitrary polynomial preprocessing, then we can also $(\alpha+o(1))$-approximate it in truly-subquadratic time without preprocessing (currently not known for any $\alpha < 3$). 
\end{description}
We remark that another related hardness result is known for the case where we only preprocess one string: Abboud and Vassilevska-Williams show that, assuming a nonuniform variant of SETH, even polynomial {\em space} (exponential time) preprocessing doesn't help to break the near-quadratic time barrier~\cite{AW19}.
\begin{description}
\item[Approximate edit-distance without preprocessing] Interestingly, using our algorithms {\em with} preprocessing (for small and large edit distance regime), we give the fastest algorithm for approximating edit distance within $3+\eps$ approximation {\em without} preprocessing. Our algorithm runs in $\tilde{O}(n^{1.6+o(1)})$ time whereas the best running time so far was $\tilde{O}(n^{1.69+o(1)})$ \cite{Andoni19}.
\end{description}
\subsection*{More context on our results}
Below we explain how the parameters in our results compare to existing literature without preprocessing.
We note that another feature of our algorithms is that they are all relatively simple. Even our most technically involved contribution, the  algorithm for general edit distance, is significantly simpler than related literature (e.g.~\cite{BR19,KS19,RSSS19}).

\paragraph{Permutation-LCS} Our $O(k \log(n))$ query time is most closely related to (and inspired by) the classic $O(n \log(n))$ for longest increasing subsequence (LIS) without preprocessing. Note that for exact computation, even after arbitrary preprocessing $\Omega(k)$ bits of communication are necessary, so our running time is tight up to the $\log(n)$ factor. Contrasting to~\cite{AN10}, we get exact result as opposed to approximation and significantly better query time bounds for $k=O(\sqrt{n})$.
\paragraph{Small edit distance}
Our $O(k^2 \log(n))$-time algorithm is most closely related to (and inspired by) a classic $\tilde{O}(n+k^2)$-time algorithm without preprocessing. Note that our near-$n^2$ SETH-lower-bound for general edit distance with preprocessing extends to $k^2$  SETH-lower-bound by a trivial padding argument (see also~\cite{BK18}). Hence our running time is near-tight assuming (BP)-SETH. Contrasting to~\cite{GKS19}, we again get exact result and better query time bound for all regimes of $k$, even when we allow single string preprocessing.

\begin{table}[ht]
	\centering
	\caption{Taxonomy of Algorithms Approximating Edit Distance}
	\label{TaxomonyTable}
	\begin{tabular}{|l|l|l|l|l|}
		\hline
		Authors        & Time  & Approximation Factor & comments                 \\ \hline
		
		\cite{CDGKS18}   & $O(n^{1.714}) $ & $3+\eps$~\footnote{The original paper~\cite{CDGKS18} reported an approximation factor of $1680$, but the authors confirmed that the approximation factor can be brought down to $3+\eps$.}& \\ \hline
		\cite{Andoni19}   & $\tilde{O}(n^{1.69})$ & $3+\eps$&
		\\ \hline
		{\bf This paper}   & {\boldmath $\tilde{O}(n^{1.6+\eps}) $} & {\boldmath $3+\eps$} &  \\ \hline
		\cite{Andoni19}   & $\tilde{O}(n^{1.5+\eps}) $ & $f_1(\delta)$ & \\ \hline
		{\bf This paper}   & {\boldmath $\tilde{O}(n^{1.5+\eps}) $} & {\boldmath $7+\eps$} & {\bf \boldmath using $ \tilde{O}(n^2)$-time preprocessing} \\ \hline
		\cite{KS19,BR19}   & $\tilde{O}(n^{1+\delta}) $ & $f_2(\delta)$ & $+ n^{1-f_3(\delta)}$ additive error\\ \hline
		
	\end{tabular}
\end{table}

\paragraph{Approximate edit distance} 
This result is most closely related to (and inspired by) recent subquadratic time approximation algorithms for edit distance~\cite{BEGHS18,CDGKS18,Andoni19,BR19,KS19}.
Here, the state of the art results include a $(3+\eps)$-approximation in $\tilde{O}(n^{12/7})$ time~\cite{CDGKS18} and later improvement to time $\tilde{O}(n^{1.69})$~\cite{Andoni19}, $f(\eps)$-approximation in time $O(n^{1.5+\eps})$~\cite{Andoni19}, or $f'(\eps)$-approximation in time $O(n^{1+\eps})$ when the true edit distance is large~\cite{BR19,KS19} (here $f,f'$ are functions that go to infinity as $\eps$ decreases).
While the improvement is not as dramatic as for sublinear algorithms, after near-quadratic preprocessing, our algorithm is clearly faster than~\cite{CDGKS18,Andoni19} ($n^{1.5}$ vs $n^{1.69}$), while obtaining much better approximation guarantees than~\cite{Andoni19,BR19,KS19} ($7+\eps$ vs $f(\eps)$). 
Interestingly, this algorithm combines ideas from aforementioned recent advances on approximate edit distance computation \cite{BEGHS18,CDGKS18,Rub18-blog,BR19}, together with our algorithm for small edit distance computation with preprocessing. Even more surprisingly, by combining ideas from our algorithms with preprocessing, we design the fastest $3+\eps$ approximation algorithm for edit distance without any preprocessing.


%



\subsection{Open problems}
We now describe a couple of exciting directions for future work

\paragraph{Preprocess one string:}
An appealing variant of our preprocessing model is when only one of the string is preprocessed. (This is motivated by a scenario where a single reference string is compared to many strings that are only used once.)
For sublinear algorithms, we are able to get some improvement over state of the art, but the $\Omega(n/k)$ lower bound from communication complexity continues to hold here.
With subquadratic algorithms on the other hand, our preprocessing algorithm has a natural variant that could be applied to only one string.
But so far we are unable to use it to obtain significant improvement over no-preprocessing approximate edit distance algorithms.
\begin{oq}
What is the complexity of approximate edit distance after preprocessing one of the strings?
\end{oq}

\paragraph{Approximate edit distance in sub-linear time}
A natural question is whether we can combine ideas from our exact $\tilde{O}(k^2)$-time algorithm for small edit distance together with the $O(n^{1.5+o(1)})$-time approximation algorithm for general edit distance to approximate small edit distance in truly sub-$k^2$ time.
Alternatively, it may be possible to show unconditional lower bounds (e.g. via communication complexity) for approximate edit distance in this regime.

\begin{oq}
What is the complexity of approximate edit distance with preprocessing when $k \ll n$?
\end{oq}

\paragraph{Beyond string alignment?}

As discussed before, preprocessing is particularly appealing when it runs in near-linear time and the queries run in sub-linear time.
In the context of string alignment, there is a very natural notion of preprocessing where each string is preprocessed separately.
An interesting, open-ended direction is to identify other problems in sub-linear algorithms where one can define preprocessing models that are both natural and allow for significant improvements.

\begin{oq} Define preprocessing models for other problems in sub-linear algorithms that are both natural and allow for significant improvements.
\end{oq}

\section{Small Ulam distance}
\label{sec:ulam}




In this section, we prove Theorem~\ref{theorem:lcs} where with preprocessing we can compute ulam distance (bounded by $k$) exactly in time $O(k\log(n))$.
\begin{theorem}[Permutation-LCS] 
\label{theorem:lcs}
Given two permutations $X,Y$ of $\{1,\dots,n\}$ with a common string of length at least $n-k$, 
we can compute their LCS exactly with $O(n\log(n))$-time preprocessing and $O(k\log(n))$-time joint processing.
\end{theorem}

\begin{claim}[Structure of close permutations] \label{claim:LIS}
If two permutations $X,Y$ of $\{1,\dots,n\}$ share a common string of length at least $n-k$,
then they can be partitioned into $O(k)$ contiguous blocks such that each block of $Y$ has an identical block in $Y$.
\end{claim}
\begin{proof}
The shared common string can be partitioned into at most $k+1$ blocks that are contiguous for $X$,
and similarly for $Y$. The coarsest refinement of both partitions is contiguous on both $X$ and $Y$ and uses at most $2k+1$ blocks.
\end{proof}

\subsubsection*{Algorithm description}
The preprocessing algorithm (Algorithm~\ref{alg:hash}) constructs $\log(n)+1$ hash tables.
The $\ell$-th hash table corresponds to window size $2^{\ell}$;
we use a rolling hash function (e.g. Rabin fingerprint) to construct a hash table 
of all contiguous substrings of $X$ of length $2^{\ell}$ in time $O(n)$.

Algorithm~\ref{alg:compress} finds the partition into blocks guaranteed in Claim~\ref{claim:LIS}.
At each iteration of the algorithm, it finds the longest contiguous substring of $X$,
starting from XStart that has an identical contiguous substring in $Y$. 
Using the prestored hashes, this is done in time $O(\log(n))$.

Finally, given the partition into blocks, we just have to solve a heaviest increasing substring problem on the $O(k)$ blocks (with weights corresponding to block lengths).
This can be done in time $O(k\log(k))$ using a standard generalization of the classic LIS algorithm (e.g.~\cite{JV92}).
We provide pseudocode in Algorithm~\ref{alg:HIS} for completeness.

In the pseudocode below we sometimes abuse notation and think of X,Y as functions from indices to characters, and similarly, we use Y$^{-1}$ to denote the inverse of this function (i.e. given a character it returns its index in $Y$.

\begin{algorithm}\label{alg:hash}\small
\caption{\small $\textsc{Preprocess}$($X$)}
$n \gets \text{length}(X)$\\
\For{\upshape $\ell = 0 \dots \log(n)$}
{
	$H$[$\ell$] $\gets$ Rolling hash of $X$ with window of length $2^{\ell}$
}
\Return $H$
\end{algorithm}

\begin{algorithm}\small
\TitleOfAlgo{\small $\textsc{Compress}$($X,H_X,Y,H_Y$)}
\caption{Algorithm \textsc{Compress} iteratively finds maximal blocks [XStart...XEnd] in X that have a matching maximal block [YStart...YEnd] in Y. At each iteration it first exponentially increases the variable $\ell$ until $\ell := \lfloor \log_2(\text{XEnd}-\text{XStart}) \rfloor$; it then binary searches for the exact length of the block.}
 \label{alg:compress}
$n \gets \text{length}(X)$\\
XStart  $\gets 0$\\
XBlocks $\gets \emptyset$\\
\While{\upshape XStart  $ < n$}
{
   YStart $\gets Y^{-1}(X(\text{XStart}))$ \tcp*{XStart,YStart = respective starts of next block}
	$\ell \gets 1$\\
	\While{\upshape $\ell < \log(n)$}
	{
		\If{\upshape $H_X$[$\ell$][XStart] $\notin H_Y$[$\ell$]}
			{\Break}
		$\ell \gets \ell +1$
	}
	XEnd $ \gets \text{XStart }+2^{\ell}$\\
	YEnd $ \gets \text{YStart }+2^{\ell}$\\
	\While{$\ell > 0$}
	{
		$\ell \gets \ell -1$\\
		\If{\upshape $H_X$[$\ell$][XEnd] == $H_Y$[$\ell$][YEnd]}
		{
			XEnd $ \gets \text{XEnd }+2^{\ell}$\\
			YEnd $ \gets \text{YEnd }+2^{\ell}$\\
		}
	}
	XBlocks $\gets \text{XBlocks} \cup (\text{XStart},\text{XEnd}-\text{XStart})$\\
	$\text{XStart} \gets \text{XEnd} +1$\\
}
\Return XBlocks
\end{algorithm}

\begin{algorithm}\label{alg:HIS}\small
\TitleOfAlgo{\small $\textsc{HIS}$(XBlocks,Y)}
\caption{Algorithm HIS maintains data structure (balanced binary search tree) Pareto, which stores the total weight and Y-index of the last character of each common substring of X and Y. The data structure is maintained sorted by Y-index, and we only keep common substrings that are {\em pareto-optimal} (in the sense that we want common substrings that are heavier but end on lower Y-index).}
$k \gets \text{length}(\text{XBlocks})$\\
Pareto $\gets$ new balanced binary search tree\\
Pareto.insert($0,0$)\\
\For{\upshape $i= 1\dots k$}
{
	\tcc{Add the next block to Pareto:}
	newY $\gets Y^{-1}(\text{Xblock[$i$].start})$\\
	prevY $\gets$ Pareto.prev(newY).Y\\
	prevWeight $\gets$ Pareto.prev(newY).weight\\
	newWeight $\gets \text{prevWeight} + \text{Xblock.weight}$\\
	Pareto.insert(newY,newWeight)\\
	\tcc{Remove old blocks that are no longer pareto-optimal:}
	\While{\upshape newWeight $\geq$ Pareto.next(newY).weight}
	{
		Pareto.next(newY).delete()
	}
}
\Return Pareto.max().weight
\end{algorithm}

\newpage
\section{Small Edit Distance}
\label{sec:small}
In this section, we prove our result on small edit distance, when the edit distance is bounded by $k$. In particular, we prove Theorem~\ref{thm:small-edit}.
\begin{theorem}[Small-EDIT]\label{thm:small-edit}
Given two strings $A=a_1a_2..a_n$ and $B=b_1b_2..b_n$ of length $n$ over alphabet $\Sigma$, and a bound on their edit distance, $\ED(A,B) \leq k$, we can compute their edit distance exactly with $O(n\log{(n)})$-time preprocessing and $O(k^2\log{(n)})$-time joint processing.
\end{theorem}
We first recall an algorithm developed in~\cite{Ukkonen85,LMS98,LV88,Myers86} that computes edit distance in $O(n+k^2)$ time.

\subsubsection*{Warm-up: An $O(n+k^2)$ algorithm for Edit Distance.} The well-known dynamic programming algorithm computes an $(n+1) \times (n+1)$ edit-distance matrix $D[0...n][0...n]$ where entry $D[i,j]$ is the edit distance, $\ED(A^i,B^j)$ between the prefixes $A[1,i]$ and $B[1,j]$ of $A$ and $B$, where $A[1,i]=a_1a_2...a_i$ and $B[1,j]=b_1b_2...b_j$. The following is well-known and easy to verify coupled with the boundary condition $D[i,0] = D[0,i] = i$ for all $i \in [0,n]$. 

For all $i, j \in [0,n]$
\[ D[i,j] = \min \left\{ \begin{array}{ll}
         D[i-1,j]+1 & \mbox{if $i > 0$};\\
        D[i,j-1]+1 & \mbox{if $j > 0$};\\
        D[i-1,j-1]+\mathds{1}(a_i\neq b_j) & \mbox{if $i,j > 0$}.
        \end{array} \right. \]

The computation cost for this dynamic programming is $O(n^2)$. To obtain a significant cost saving when $\ED(A,B) \leq k <<n$, the $O(n+k^2)$ algorithm works as follows. It computes the entries of $D$ in a greedy order, computing first the entries with value 0, $1,2,...k$ respectively. Let diagonal $d$ of matrix $D$, denotes all $D[i,j]$ such that $j=i+d$. Therefore, the entries with values in $[0,k]$ are located within diagonals $[-k,k]$.  Now since the entries in each diagonal of $D$ are non-decreasing, it is enough to identify for every $d \in [-k,k]$, and for all $h \in [0,k]$, the last entry of diagonal $d$ with value $h$. The rest of the entries can be inferred automatically. Hence, we are overall interested in identifying at most $(2k+1)*k$ such points. 
The $O(n+k^2)$ algorithm shows how building a suffix tree over a combined string $A\$B$ (where $\$$ is a special symbol not in $\Sigma$) helps identify each of these points in $O(1)$ time, thus achieving the desired time complexity.

Let $L^h(d)=\max\{i : D[i,i+d]=h\}$. The $h$-wave is defined by $L^h=\langle L^h(-k),...,L^h(k) \rangle$. Therefore, the algorithm computes $L^h$ for $h=0,..k$ in the increasing order of $h$ until a wave $e$ is computed such that $L^e(0)=n$ (in that case $\ED(A,B) = e$), or the wave $L^k$ is computed in the case the algorithm is thresholded by $k$. Given $L^{h-1}$, we can compute $L^h$ as follows. 

Define $$Equal(i,d)=\max_{q \geq i}{(q \mid A[i,q]=B[i+d,q])}$$
Then, $L^0(0)=Equal(0,0)$ and 
\[ L^h(d) = \max \left\{ \begin{array}{ll}
         Equal(L^{h-1}(d)+1,d) & \mbox{if $h-1 \geq 0$};\\
        Equal(L^{h-1}(d-1),d) & \mbox{if $d-1 \geq -k , h-1 \geq 0$};\\
        Equal(L^{h-1}(d+1)+1,d) & \mbox{if $d+1, h+1 \leq k$}.
        \end{array} \right. \] 

Using a suffix tree of the combined string $A\$B$, any $Equal(i,d)$ query can be answered in $O(1)$ time. Next, we show that it is possible to preprocess each $A$ and $B$ separately so that even then each $Equal(i,d)$ query can be implemented in $O(\log{n})$ time. 

\subsubsection*{Preprocessing Algorithm}
\label{sec:small-preprocess}

The preprocessing algorithm (Algorithm~\ref{alg:hash}) constructs $\log(n)+1$ hash tables just like in Section~\ref{sec:ulam}.
The $\ell$-th hash table corresponds to window size $2^{\ell}$;
we use a rolling hash function (e.g. Rabin fingerprint) to construct a hash table 
of all contiguous substrings of $X$ of length $2^{\ell}$ in time $O(n)$. Since there are $\log{n}+1$ levels, the overall preprocessing time is $O(n\log{n})$. Let $H_A[\ell]$ store all the hashes for windows of length $2^\ell$ of $A$ and similarly $H_B[\ell]$ stores all the hashes for windows of length $2^\ell$ of $B$.

\subsubsection*{Answering $Equal(i,d)$ in $O(\log{n})$ time}
\label{sec:small-query}
$Equal(i,d)$ queries can be implemented by doing a simple binary search over the presorted hashes in $O(\log{n})$ time. The pseudocode is given below. Suppose $Equal(i,d)=q$. The first While loop (line $5$-$8$) identifies the smallest $\ell \geq 0$ such that $q < 2^{\ell}$. The next While loop does a binary search for $q$ between $i+2^{\ell-1}$ to $i+2^{\ell}$.

\begin{algorithm}\label{alg:reduction}\small
\caption{\small $\textsc{Equal}$($i,d, A, H_A,B,H_B$)}
$n \gets \text{length}(A)$\\
AStart  $\gets i$, BStart $\gets i+d$\\

	$\ell \gets 0$\\
	\While{\upshape $\ell < \log(n)$}
	{
		\If{\upshape $H_A$[$\ell$][AStart] $\neq H_B$[$\ell$][BStart]}
			{\Break}
		$\ell \gets \ell +1$
	}
	AStart $\gets i+2^{\ell-1}$, BStart $\gets i+d+2^{\ell-1}$\\
	AEnd $ \gets i+2^{\ell}-1$, BEnd $ \gets i+d+2^{\ell}-1$\\
	Mid $\gets \frac{(\text{AEnd}-\text{AStart}+1)}{2}$\\
	\While{$Mid \geq 1$}
	{
		
		\If{\upshape $H_A$[Mid][AStart] == $H_B$[Mid][BStart]}
		{
			AStart $ \gets \text{AStart}+Mid$, BStart $ \gets \text{BStart }+Mid$\\
		}
		\Else
		{
			AEnd $ \gets \text{AEnd}-Mid-1$, BEnd $ \gets \text{BEnd }-Mid-1$\\
		}
		Mid $\gets \frac{(\text{AEnd}-\text{AStart}+1)}{2}$\\

}
\Return AEnd
\end{algorithm}

Implementing $Equal(i,d)$ query in $O(\log{n})$ time together with the correctness proof of $O(n+k^2)$ algorithm leads to Theorem~\ref{thm:small-edit}.

\section{Preprocessing a Single String: Answering Gap Edit Distance in Sublinear Time}
\label{sec:one-string}
In this section, we design an algorithm that given two strings $A$ and $B$, preprocess only one string, say $B$. During the query phase, the string $A$ is provided, and a query algorithm must answer whether $\ED(A,B) \leq k$ or $\ED(A,B) \geq 2k^2$. We give an algorithm for this quadratic gap-edit distance problem that runs in $\tilde{O}(\frac{n}{k}+k^2)$ time. Therefore, the algorithm achieves a sublinear query time whenever $k \leq n^{1/2}$ and $k \geq $poly$\log{n}$. Note that this problem was recently studied in \cite{GKS19} without any preprocessing. They achieve a running time bound of $\tilde{O}(\frac{n}{k}+k^3)$.

\subsection{Preprocessing Algorithm}
Given $Y \in \Sigma^n$, we sample each index in $[1,n]$ uniformly at random with probability $\frac{\log^2{n}}{k}$. Let $S=\{i_1,i_2,...,i_s\}$ denote the sampled indices. Create the following substrings
\[ B^d= b_{i_1+d}b_{i_2+d}...b_{i_s+d}, \forall d=-k, -k+1,..,0,....,k-1, k.\]

By a standard application of the Chernoff bound, we can assume with probability at least  $1-\frac{1}{n^3}$, the number of sampled indices $s=\Theta(\frac{n\log^2{n}}{k})$.

The preprocessing algorithm constructs $\log(s)+1$ hash tables just like in Section~\ref{sec:small-preprocess}, but for each $B^d$, $d \in [-k,k]$.
The $\ell$-th hash table corresponds to window size $2^{\ell}$ of $B^d$. Since there are $\log{s}+1$ levels, the overall preprocessing time  is $O(k*\frac{n}{k}\log^2{n}\log{s})$=$\tilde{O}(n)$ with probability $1-\frac{1}{n^3}$. Let $H_B^d[\ell]$ store all the hashes for windows of length $2^\ell$ of $B^d$ for $d=[-k,k]$.

\sloppy
\subsection{Query Algorithm}
Given $A \in \Sigma^n$. We create a sampled substring $A_S=a_{i_1}a_{i_2}....a_{i_s}$.  We construct $\log(s)+1$ hash tables for $A_S$. Again, the $\ell$-th hash table corresponds to window size $2^{\ell}$ of $A_S$. Since there are $\log{s}+1$ levels, the overall time to compute the hashes of $A_S$ is $O(\frac{n}{k}\log^2{n}\log{s})=\tilde{O}(\frac{n}{k})$ with probability $1-\frac{1}{n^3}$. Let $H_A[\ell]$ store all the hashes for windows of length $2^\ell$ of $A_S$

We now define an approximate $Equal(i,d)$, $Approx\text{-}Equal(i,d)$ query as follows. Let $n(i) \geq i$ be the nearest index to $i$ present in $S$. Define
\[Approx\text{-}Equal(i,d)=\max_{q \geq i}{(q \mid q \in S, X_S[n(i),q]=Y^d[n(i),q])}\]

We now run the same algorithm from Section~\ref{sec:small-query} except that we replace $Equal(i,d)$ with $Approx\text{-}Equal(i,d)$. Let us use $\hat{L}^h$ to denote the $h$-wave computed by using $Approx\text{-}Equal(i,d)$ for $h \in [0,k]$ and $d \in [-k,k]$. If the algorithm computes $\hat{L}^h(0) \geq n$ for $h \leq k$, the algorithm returns YES. Else, it returns NO.

Clearly, the running time of the algorithm is $\tilde{O}(\frac{n}{k}+k^2)$. We now show that the algorithm solves the quadratic gap problem. 

\paragraph*{Analysis}
When comparing a symbol $x_i$ with $y_j$, if they do not match, we call it a 'mismatch'. The following is an easy lemma which shows we cannot miss too many mismatches due to sampling.

\begin{lemma}
\label{claim:1}
Given $i \in [1,n]$ and $d \in [-k,k]$, let $i' \geq i$ be the smallest index  such that $x_ix_{i+1}...,x_{i'}$ and $y_{i+d}y_{i+1+d}...y_{i'+d}$ have at least $k'=\frac{k}{\log{n}}$ mismatches. Let $i \leq j_1,j_2,...,j_{k'} = i'$ be the indices such that $x_{j_h} \neq y_{j_h}+d$. Define a bad event $B(i,d)$ to be the event that none of these $k'$ mismatch indices are sampled. Then $Prob(Bad(i,d)) \leq 1-\frac{1}{n^3}$. Moreover, all bad events are avoided with probability at least $1-\frac{1}{n}$.
\end{lemma}
\begin{proof}
Since the sampling probability is $\Theta(\frac{\log^2{n}}{k})$, the expected number of points sampled from $j_1,j_2,..,j_{k'}$ is $\Theta(\log{n})$. Now, by the Chernoff bound, the probability that none of them are sampled can be made to be $1-\frac{1}{n^3}$ (by choosing the constants in the sampling probability appropriately).

Then by a union bound over all $i \in [1,n]$ and $d \in [-k,k]$, with probability $\geq 1-\frac{1}{n}$ none of the bad events $Bad(i,d)$ happen.
\end{proof}

Therefore, we can assume all bad events are avoided. The above lemma leads to the following direct corollary.

\begin{corollary}
\label{claim:2}
For all $i \in S$, $\ell \in [0,\log{s}+1]$ and $d \in [-k,k]$ if $H_A[\ell][i]=H_B^d[\ell][i]$ then $a_ia_{i+1}...,a_{i+2^{\ell}}$ and $b_{i+d}b_{i+1+d}...b_{i+d+2^{\ell}}$ have less than $\frac{k}{\log{n}}$ mismatches. 
\end{corollary}
\begin{proof}
Take any $i$ and $d$. Since $Bad(i,d)$ did not happen, if $a_ia_{i+1}...,a_{i+2^{\ell}}$ and $b_{i+d}b_{i+1+d}...b_{i+d+2^{\ell}}$ had at least $\frac{k}{\log{n}}$ mismatches, we would have $H_A[\ell][i]\neq H_B^d[\ell][i]$
\end{proof}
Using the above corollary, we can now show that $Approx\text{-}Equal(i,d)$ is a good approximation of $Equal(i,d)$.
\begin{lemma}
\label{claim:3}
If $Approx\text{-}Equal(i,d)=q$ then $a_ia_{i+1}....a_{q}$ and $b_{i+d}b_{i+d+1}...b_{q+d}$ have strictly less than $2k$ mismatches.
\end{lemma}
\begin{proof}
Since the sampling probability is $\frac{\log^2{n}}{k}$, $(n(i)-i) \leq \frac{k}{\log{n}}$ with high probability (we assume $k \geq$ poly$\log{n}$).

Now $A[n(i),q]$ can be decomposed into at most $\log{s}+1 \leq \log{n} +1$ intervals each of length that is a power of two. Moreover for each of these intervals the computed hashes $H_A$ and $H_B^d$ must match. Therefore, each of these at most $\log{n}+1$ intervals can have at most $\frac{k}{\log{n}}$ mismatches from Corollary~\ref{claim:2}. Thus the total number of mismatches is strictly less than $(n(i)-i)+(\log{n}+1)\frac{k}{\log{n}}=k+\frac{2k}{\log{n}} \leq 2k$.
\end{proof}

In order to complete our analysis, we now compare the $h$-waves computed by the exact algorithm from Section~\ref{sec:small-preprocess} and approximate $h$-waves computed by using $Approx\text{-}Equal(i,d)$. 


\begin{lemma}[Completeness]
\label{claim:completeness}
$\forall h \in [0,k]$ and $d \in [-k,k]$, $\hat{L}^h(d) \geq L^h(d)$. Therefore, if $\ED(A,B) \leq k$, then the algorithm will return YES.
\end{lemma}
\begin{proof}
The proof follows simply by induction since $Approx\text{-}Equal(i,d) \geq Equal(i,d)$.
\end{proof}
\begin{lemma}[Soundness]
\label{claim:soundness}
$\forall h \in [0,k]$ and $d \in [-k,k]$, $\hat{L}^h(d) \leq L^{2k(h+1)}(d)$. Therefore, if $\ED(A,B) > 2k^2+2k$, then the algorithm will return NO.
\end{lemma}
\begin{proof}
The proof is again by induction. Observe that $\hat{L}^0(0)=Approx\text{-}Equal(0,0)$. From Lemma~\ref{claim:3}, $B[0,\hat{L}^0(0)]$ and $B[0,\hat{L}^0(0)]$ can have at most $2k$ mismatches. Therefore, $L^{2k}(0) \geq \hat{L}^0(0)$. Suppose the result is true for $0,1,..,h-1$ for all diagonals $\in [-k,k]$ and upto diagonal $d-1$ for $h$. Let us consider $\hat{L}^h(d)$. Recall the definition of  $\hat{L}^h(d)$. 
\[ \hat{L}^h(d) = \max \left\{ \begin{array}{ll}
         Approx\text{-}Equal(\hat{L}^{h-1}(d)+1,d) & \mbox{if $h-1 \geq 0$};\\
        Approx\text{-}Equal(\hat{L}^{h-1}(d-1),d) & \mbox{if $d-1 \geq -k , h-1 \geq 0$};\\
        Approx\text{-}Equal(\hat{L}^{h-1}(d+1)+1,d) & \mbox{if $d+1, h+1 \leq k$}.
        \end{array} \right. \] 
Let us consider the first expression, $Approx\text{-}Equal(\hat{L}^{h-1}(d)+1,d)$. By the induction hypothesis, $\hat{L}^{h-1}(d) \leq L^{2kh}(d)$. There must be a mismatch at row $\hat{L}^{h-1}(d)+1$. The number of mismatches in $Approx\text{-}Equal(\hat{L}^{h-1}(d)+1,d)$ is at most $2k-1$ mismatches. Therefore, $L^{2k(h+1)}(d) \geq Approx\text{-}Equal(\hat{L}^{h-1}(d)+1,d)$. Similarly, for the other two expressions. 

Therefore, if $\ED(A,B) > 2k^2+2k$, then $L^{2k^2+2k}(0)< n$. Then $\hat{L}^k(0) \leq L^{2k^2+2k}(0)< n$, the algorithm aborts and declares NO.
\end{proof}

Hence, we get the following theorem.
\begin{theorem}[Small-EDIT-Single-Preprocessing]\label{thm:small-edit-2}
Given two strings $A=a_1a_2..a_n$ and, $B=b_1b_2...b_n$ of length $n$ over alphabet $\Sigma$, we can answer if $\ED(A,B) \leq k$ or $\ED(X,Y) > 3k^2$ with probability at least $1-\frac{1}{n}$ by preprocessing only a single string in $\tilde{O}(n)$-time and with a query time of $\tilde{O}(\frac{n}{k}+k^2)$.
\end{theorem}

\section{Large edit distance, $7+\eps$-approx}\label{sec:approx}

In this section we prove our result for the large edit distance regime. Our main result is a $7+o(1)$ approximation for $\ED(A,B)$ in $n^{\frac{3}{2}+o(1)}$ query time. We are allowed to preprocess each $A$ and $B$ separately and spend $\sim n^2$ time in overall preprocessing.
\begin{remark}[Estimating the distance vs computing an alignment] 
For simplicity of presentation, we write our algorithms as merely estimating the distance. It is straightforward with standard techniques to modify them to output the alignment as well in roughly the same running time.
\end{remark}

\paragraph{Organization of this section}
In Subsection~\ref{sub:high-level} we give a bird's eye overview of the main technical elements of our algorithm.
Subsection~\ref{sub:windows} formally describes the decomposition of the strings into windows, Subsection~\ref{sub:reduction} is a standard dynamic programming for computing an optimal window-compatible matching from pairwise distances. 
Our main contribution is in Subsection~\ref{subsec:algo} which describes the algorithm for learning the close-window graph.

\subsection{High Level Description of the Algorithm} \label{sub:high-level}

\subsubsection*{The basic divide-and-conquer framework for approximate edit distance}
The algorithm builds upon the recent progress on approximating edit distance in subquadratic time using divide-and-conquer algorithms \cite{BEGHS18,CDGKS18,Andoni19,BR19,KS19,RSSS19}, along with our small-edit-distance algorithm from Section~\ref{sec:small}. We decompose the strings $A$ and $B$ into contiguous substrings called {\em windows}. These windows can be overlapping and have variable lengths. 
Up to an $(1+o(1))$-factor approximation, we can now wlog restrict our attention to matchings of $A$ to $B$ that are ``window-compatible'', i.e. they respect the partition to windows (see Lemma~\ref{lem:windowComptabile}).

If we (approximately) knew all the pairwise distances between windows, a standard DP would find an (approximately) optimal window-compatible matching efficiently (Lemma~\ref{lem:DP}).
Computing the pairwise distances is further reduced to (approximately) learning the bipartite {\em close-window graph}, where a pair of $A$- and $B$-windows are neighbors if their pairwise edit distance is below an appropriate threshold $\tau$.

The goal is now to approximately learn the close-window graph while computing as few window-window distances as possible. With this in mind, we classify the windows as either {\em dense} (high-degree in the close window graph), or {\em sparse}. 
We use by-now-standard separate subroutines to handle each kind of windows.

\subsubsection*{Further details of our algorithm}
The density of a window can be estimated by computing its edit distance to a small sample of its potential neighbors. To obtain optimal tradeoff between parameters, we cannot afford even this small sample to classify windows as dense or sparse. Here we deviate from previous works and estimate the density on-the-fly. 
That is each window is assumed to be sparse by default, and only when it is selected as a special ``seed'' for the sparse subroutine, we estimate its degree and move it to the dense subroutine if necessary. 
(In fact, an originally dense window can lose many of its neighbors and become sparse by the time it is selected; this does not hurt our analysis.)

The main sparse subroutine proceeds by recursively narrowing down the set of relevant candidate neighbors. Even though sparse windows take part in multiple levels of recursion, the loss in approximation from each level of the sparse subroutine is negligible, so it continues to be negligible in aggregate. The dense subroutine incurs the main loss in approximation due to the use of triangle inequality. Fortunately, each dense window can only contributes to one level of the entire recursion and thus the overall approximation factor remains bounded.

When we compute the edit distance between pairs of windows, we do it  exactly using our algorithm from Section~\ref{sec:small}. This algorithm is very efficient when the windows are close, but its running time may be as slow as quadratic in the window size when the distance is large.
We remark that three recent approximate edit distance algorithms~\cite{Andoni19,BR19,KS19} also use the basic divide-and-conquer framework, yet manage to obtain comparable or faster running times without preprocessing. 
Those algorithms compute window-window distances by recursively applying an approximate edit distance algorithm; while this improves efficiency, the approximation factor explodes exponentially in the depth of the recursion.



\subsection{Decomposition into Variable Sized Windows}
\label{sub:windows}

\paragraph{Parameters Settings.}

We divide the strings $A,B$ into \textit{windows}, equivalently  contiguous substrings.  
We use $d$ and $t$ to denote the window width and the number of windows of $A$ respectively. Fix $d=n^{1/4}$ and $t=\frac{n}{d}=n^{3/4}$ throughout the presentation.  

Let $\eps>0$ be an arbitrarily small constant (or slightly sub-constant), such that we would like to obtain a $(7+O(\eps))$-approximation in $O(n^{3/2+O(\eps)})$-time.
The windows in $B$ will vary in width. Moreover, they can be overlapping where the amount of overlap will be controlled by a parameter $\tau$ which is the relative $\ED$ threshold between a pair of windows. We will vary $\tau$ geometrically, and for each value of $\tau$, we will compute a set of windows $\cB^{\tau}$. Let $\ttau$ denote the number of windows of $\cB^{\tau}$. We will have $\ttau=O(\frac{n}{\eps\tau d})$.

\paragraph{Choice of Windows.}
The choice of windows play a crucial role in our overall algorithm design. For the string $A$,  partition $A$ into disjoint {\em windows} of width $d$ denoted by $\cA$. 

\[\cA=\set{A[1,d],A[d+1,2d],\dots, A[n-d+1,n]}\]

We now compute the windows of $B$. Let us take $\tau=\{0,\frac{1}{d}, \frac{(1+\eps)}{d}, \frac{(1+\eps)^2}{d},...,1\}$. For each value of $\tau$, we compute a set of windows $\cB^{\tau}$. Finally, we set $\cB=\cup_{\tau} \cB^{\tau}$ to denote all  computed windows of $B$.

For $\tau=0$, take $h_{\tau}=d$, and $l_{\tau}=d$. For $\tau=\frac{1}{d}$, take $h_{\tau}=d+1, l_{\tau}=d-1$. In general, for  $\tau = \frac{(1+\epsilon)^{j}}{d}$, $j \ge 1$, take $h_{\tau}=\lfloor{d+(1+\eps)^{j-1}\rfloor}$ and $l_{\tau}=\lfloor{d-(1+\eps)^{j-1}\rfloor}$.

Set $\gamma_{\tau}=\max{(1, \lfloor \epsilon \tau d \rfloor)}$. Define

\[\cH^{\mathsf{\tau}}:=\set{B[1,h_{\tau}),B[ \gamma_\tau+1, \gamma_\tau+h_{\tau}), B[2\gamma_\tau+1, 2\gamma_\tau+h_{\tau}],....} \]

\[\cL^{\mathsf{\tau}}:=\set{B[1,l_{\tau}),B[ \gamma_\tau+1, \gamma_\tau+l_{\tau}), B[2\gamma_\tau+1, 2\gamma_\tau+l_{\tau}],...} \]

Finally, $\cB{\mathsf{\tau}}=\cH^{\mathsf{\tau}} \cup \cL^{\mathsf{\tau}}$, that is $\cB{\mathsf{\tau}}$ consists of intervals of length $h_\tau$ and $l_\tau$ starting at every $\gamma_\tau$ grid points.

For window $a\in \cA$ (similarly for windows in $\cB$), let $s(a)$ denote the starting index of $a$ (e.g., $s(A[1,d]) = 1$) and let $e(a)$ denote index the the last index of $a$ ($e(A[1,d]) = d$). This completes the description of the windows.

Note that overall we create $t=\frac{n}{d}$ windows of $A$ and $\ttau=O(\frac{n}{\gamma_\tau})$.

\paragraph{Mapping between windows.}

We say that a mapping $\mu :\cA \to  \cB\cup \set{\perp}$ between windows is monotone if for all $a,a'\in  \cA$ such that 
$\mu (a), \mu(a')\neq \perp$ and $s(a)\le s(a')$ 
we also have that  $s(\mu (a))\le s(\mu (a'))$ and $e(\mu (a))\le e(\mu (a'))$.
Setting $\mu (a) =\perp$ represents deleting $a$ from the string.  
As such, we define $ED(a,\perp) =d$ for all windows $a$.

By abuse of notation, we let $\mu\subset \cA$ denote the set of $A$-windows such that $\mu(a)\neq \perp$.
For $a\in \mu$, let $\anext$ denote the window $a' \in \mu$ immediately after $a$ (note that next depends on the mapping $\mu$).  
If $a$ is the last window in $\mu$, we define $\anext :=\perp$.  We define $\aprev$ in the analogous way.

For a monotone mapping $\mu$ we define its \textit {edit distance} as:

\begin{align*} \ED(\mu)&=\sum_{a\in \cA}\ED(a,\mu(a))  
+ \sum_{i=1}^n \Big|\#\{a \text{~s.t.~}i \in \mu(a)\} -1\Big|
\end{align*}
The first term is just sum of the edit distances between matched windows.
To understand the second term, notice that for each $i$ we expect it to appear in the image $\mu(a)$ of exactly $1$ window. The second term sums the difference between the number of appearances of $i$ and $1$; it is a penalty for either overlap of windows (requiring deletions) or excessive spacing (requiring insertions).

The next lemma asserts that the cost of a minimal monotone mappings provides a good approximation for the 
actual edit distance between the input strings $A,B$.

\begin{lemma}\label{lem:windowComptabile}
	Let $A,B\in \Sigma^n$, then the following holds:
\begin{enumerate}
	\item For every monotone mapping $\mu:\cA\to \cB\cup \set{\perp}$ we have: $\ED(\mu)\ge \ED(A,B)$.
	\item There exists a monotone mapping $\mu:\cA\to \cB\cup \set{\perp}$ satisfying: $ED(\mu)\le (1+8\eps) \ED(A,B).$
\end{enumerate}
\end{lemma}

\begin{proof}

(Part 1.) Using $\mu$ we construct an explicit mapping from $A$ to $B$. For each window, $a \in \cA$, transform the characters of $a$ in $A$ into $\mu(a)$. If $\mu(a)=\perp$, then delete all the characters of $a$. Finally, if $a$ is not the last window, $\mu(a)\neq \perp$ and $e(\mu(a)) - s(\mu(\anext)) > 0$ then delete that many characters from the end of $\mu(a)$.

Let $\mu(A)$ be the currently transformed string. Then, by construction
\[ \ED(A,\mu(A))\leq \sum_{a \in \cA} \ED(a,\mu(a)) + \sum_{i:\#\{a \text{~s.t.~}i \in \mu(a)\} \geq 1 }  [\#\{a \text{~s.t.~}i \in \mu(a)\} -1]\]
Here the second term counts for all indices of $B$ that are counted more than once in $\mu$.

Furthermore, since we deleted any overlaps between $\mu(a)$s, we have $\mu(A)$ as a subsequence of $B$. Thus,
\[ \ED(\mu(A),B)=|B-\mu(A)|=\sum_{i:\#\{a \text{~s.t.~}i \in \mu(a)\} =0 }  [1-\#\{a \text{~s.t.~}i \in \mu(a)\} ]\]
Here we account for all indices of $B$ that are not counted in $\mu$.

Putting the above two inequalities together and using the triangle inequality, we get
\[\ED(A,B) \leq \ED(A,\mu(A))+ \ED(\mu(A),B) = \ED(\mu)\]

(Part 2.) Consider the optimal sequence of edits from $A$ to $B$. This can be viewed as $\ell$ substitutions
of characters of $A$, $k$ deletions of characters, and then $k$ insertions, where $2k+\ell=\ED(A,B)$. Let
$A'$ be the subsequence of characters of $A$ which are either untouched or substituted to match a character in $B$. Let $B'$ be the corresponding subsequence of
$B$. Let $\mu' : A' \rightarrow B'$ be the monotone correspondence between the characters of these solution. 

We now construct a mapping $\hat{\mu}: \cA \rightarrow \cB \cup \{\perp\}$ that will have low cost. $\hat{\mu}$ may not be monotone. Finally, we will convert $\hat{\mu}$ into a monotone mapping $\mu$ by negligibly increasing the cost, thus overall getting a good mapping. For each $a \in \cA$, if $a \cap A'=\emptyset$, then set $\hat{\mu}(a)=\perp$. 
For each window $a \in \cA$ which nontrivially intersects $A'$, let $i_a$ be the first index of $\mu'(a \cap A')$. Similarly, let 
$j_a$ be the last index of $\mu'(a \cap A')$. Note that $i_{\anext} > j.a$ for all $a$ with $\hat{\mu}(a) \neq \perp$ . 

Let $k_a$ denote the number of indices of $a$ that are deleted and $\ell_a$ denote the number of indices of $a$ that are substituted. Let $m_a$ denote the number of indices in $B$ between $i_a$ and $j_a$ that are not mapped from any indices of $A$ in $a \cap A'$ by $\hat{\mu}$. Then, it must hold that if $(j_a-i_a+1) \geq d$ then $m_a \geq [(j_a-i_a+1)-d]+k_a$. If $(j_a-i_a+1) < d$ then $k_a \geq [d-(j_a-i_a+1)]+m_a$.

When $(j_a-i_a+1) \geq d$, if $\frac{(1+\eps)^{j-1}}{d} \leq [(j_a-i_a+1) -d ] < \frac{(1+\eps)^{j}}{d}$ (note that such a $j$ always exists). Set $\hat{\mu}(a)$ to be the right-most interval in $\cH^{\tau}$, $\tau=\frac{(1+\eps)^{j-1}}{d}$, that contains $\mu'(A[i])$. In that case, the length of the interval $\hat{\mu}(a)$ is $d+(1+\eps)^{j-1}$.

On the other hand, if $(j_a-i_a+1) < d$ and $\frac{(1+\eps)^{j-1}}{d} \leq [d-(j_a-i_a+1)] <\frac{(1+\eps)^{j}}{d}$ (note that such a $j$ always exists). Set $\hat{\mu}(a)$ to be the right-most interval in $\cL^{\tau}$ that contains $\mu'(A[i])$. In that case, the length of the interval $\hat{\mu}(a)$ is $d-(1+\eps)^{j-1}$.


Note that $(i_a-s(\hat{\mu}(a)))\leq \eps\tau d$. We get

\begin{align*}\ED(a,\hat{\mu}(a)) & \leq \underbrace{m_a+k_a+l_a}_\text{True edits}+\underbrace{\eps(1+\eps)^{j-1}}_\text{Due to error in length estimation} + \underbrace{2(s(\mu(a))-i_a)}_\text{Shift due to the choice of grid points}\\
&\leq (1+3\eps) \tau d.
\end{align*}

Now note that, for all $a \in \cal A$, $j_a < i_{\anext}$. Let $r_a=(i_a-s(\mu(a)))$. Thus, we have $|e(\hat{\mu}(a))-s(\hat{\mu}(\anext))| \leq |r_a-r_{\anext}| + M(j_a, i_{\anext})$ where $M(j_a, i_{\anext})$ are the number of symbols of $B$ in between $j_a$ and $i_{\anext}$. Note that, in $\ED(A,B)$, all these symbols, $M(j_a, i_{\anext})$, are deleted from $B$ or equivalently inserted in $A$. If $r_a \geq r_{\anext}$ then we charge $|r_a-r_{\anext}|$ to $a$. Else, we charge it to $\anext$. Note that when $|r_a-r_{\anext}|$ is charged to $a$ then $|r_a-r_{\anext}| \leq r_a \leq \eps (m_a+k_a+l_a)$.


Hence we obtain
\begin{align*}
\ED(\hat{\mu}) &=\sum_{a\in \cA}\ED(a,\hat{\mu}(a))  
+ \sum_{i=1}^n \Big|\#\{a \text{~s.t.~}i \in \hat{\mu}(a)\} -1\Big|\\
&\leq (1+3\eps) \ED(A,B)+\sum_{a, \hat{\mu}(a) \neq \perp} 4r_a\\
&\leq (1+7\eps) \ED(A,B)
\end{align*}

It is possible that $\hat{\mu}$ may not be monotone. In particular, this can happen if $r_{\anext} > r_{a}$ so that $s(\hat{\mu}(\anext))$ comes before $s(\hat{\mu}(a))$. In this situation, we simple set $\mu(a)=\perp$ by paying at most $r_{\anext}$. In more details, we start with $\mu=\hat{\mu}$. We iterate over $a \in cA$ in increasing order of $s(a)$, whenever we encounter a window $a$ such that $s(\hat{\mu}(a)) < s(\hat{\mu}(\aprev))$, we find all $a'$ with $s(a') < s(a)$ but $s(\hat{\mu}(a)) < s(\hat{\mu}(a'))$, we set $\mu(a')=\perp$ by paying an extra edit cost of $r_{a}$.

Thus $$\ED(\mu) \leq \ED(\hat{\mu})+\sum_{a, \hat{\mu}(a) \neq \perp} r_a \leq (1+8\eps) \ED(A,B).$$

This completes the proof.
\end{proof}
\paragraph{Low-Skew Mapping.}
A monotone mapping $\mu: \cA \to \cB$ is said to have skew at most $D$ if for all $a, a'\in \cA$ we have:
$$ \frac{1}{D} \abs{s(a)-s(a')} \leq \abs{s(\mu(a))-s(\mu(a'))} \leq D \abs{s(a)-s(a')} $$
We next show that any monotone mapping $\mu$ can be transformed into a low-skew mapping with $D=\frac{1}{\eps}$ with negligible loss. 
Along with Lemma~\ref{lem:windowComptabile}, this ensures there exists a near-optimal low-skew mapping, which we will exploit in our algorithm design.

\begin{lemma}
\label{lemma:low-skew}
	For every monotone mapping $\mu:\cA \to \cB \cup \set{\perp}$, and for every $\eps>0$, 
	there exists a monotone mapping $\mu':\cA \to \cB \cup \set{\perp}$ such that:
	$\ED(\mu')\leq(1+2\eps)\ED(\mu)$ and $\mu'$ has a skew that is at most $1/\eps$.
\end{lemma}
\begin{proof}
Consider $D=\frac{1}{\eps}$. Let $S \subset \cA \times \cA$ be the set of all
pairs $(a_1,a_2)$ which have more than $\frac{1}{\eps}$ skew,  and let $s(a_1) \leq s(a_2)$. We put a partial ordering on $S$ such that
$(a_1, a_2) \preccurlyeq (a_3, a_4)$ if $s(a_3) \leq s(a_1) \leq s(a_2) \leq s(a_4)$. Let $S' \subseteq S$ be the set of pairs that are maximal with respect to this relation.

We also say $(a_1, a_2)$ is disjoint from $(a_3, a_4)$ if either $s(a_2) \leq s(a_3)$ or $s(a_4) \leq s(a_1)$. We now build a set of pair-wise disjoint elements from $S$  by repeatedly picking a maximal element and discarding all elements in $S$ that intersects it. Let us call this set $T$. Label the elements of $T$ by $(a_1,a_2), (a_3,a_4),....,(a_{2k-1},a_{2k})$ such that $s(a_1) \leq s(a_2) \leq \cdots \leq s(a_{2k-1}) \leq s(a_{2k})$. For all windows $a$ such that $s(a_{2i-1}) \leq a < s(a_{2i})$, set $\mu'(a)=\perp$. Otherwise, keep $\mu'(a)=\mu(a)$. Note that $\mu'$ has skew at most $\frac{1}{\eps}$ since for each $(a,a') \in S$, either $\mu'(a)=\perp$  or $\mu'(a')=\perp$.
\sloppy
We now show that $\ED(\mu') \leq (1+\eps) \ED(\mu)$. Pick $(a_1,a_2) \in T$. If $|s(\mu(a_1))-s(\mu(a_2))| \geq \frac{1}{\eps} |s(a_1)-s(a_2)|$, then 
\begin{align*}
\ED(A[s(a_1), s(a_2)), B[\mu(s(a_1)), \mu(s(a_2))) &\geq |s(\mu(a_1))-s(\mu(a_2))|-|s(a_1)-s(a_2)| \\
&\geq |s(a_1)-s(a_2)|(\frac{1}{\eps}-1)\end{align*}
On the other hand, by setting for every $a \in \cA$ with $s(a_1) \leq s(a) < s(a_2)$, $\mu'(a)=\perp$, we pay an edit cost of at most 
\[ \ED(A[s(a_1), s(a_2)), B[\mu(s(a_1)), \mu(s(a_2)))+ |s(a_1)-s(a_2)| \leq (1+\frac{\eps}{1-\eps})\ED(A[s(a_1), s(a_2)), B[\mu(s(a_1)), \mu(s(a_2))).\] Note that the edit cost of $\mu$ restricted to the substring $A[s(a_1), s(a_2))$ is only higher than $\ED(A[s(a_1), s(a_2)), B[\mu(s(a_1)), \mu(s(a_2)))$ due to possible double-counting indices which are part of consecutive windows within the interval $[s(a_1), s(a_2))$ that overlap under $\mu$. 

If $|s(\mu(a_1))-s(\mu(a_2))| \geq \eps |s(a_1)-s(a_2)|$, then again, we have
\begin{align*}
\ED(A[s(a_1), s(a_2)), B[\mu(s(a_1)), \mu(s(a_2))) &\geq |s(a_1)-s(a_2)|-|s(\mu(a_1))-s(\mu(a_2))| \\
&\geq (1-\eps) |s(a_1)-s(a_2)|.\end{align*}
On the other hand, by setting for every $a \in \cA$ with $s(a_1) \leq s(a) < s(a_2)$, $\mu'(a)=\perp$, we pay an edit cost of at most 
\begin{align*}
&\ED(A[s(a_1), s(a_2)), B[\mu(s(a_1)), \mu(s(a_2)))+|s(\mu(a_1))-s(\mu(a_2))| \\
&\leq \ED(A[s(a_1), s(a_2)), B[\mu(s(a_1)), \mu(s(a_2)))+\eps |s(a_1)-s(a_2)| \\
&\leq (1+\frac{\eps}{1-\eps})\ED(A[s(a_1), s(a_2)), B[\mu(s(a_1)), \mu(s(a_2))).
\end{align*}

Now, going over all disjoint $(a_{2i-1}, a_{2i}) \in T$, we get for $\eps \leq \frac{1}{2}$
\[\ED(\mu') \leq (1+\frac{\eps}{(1-\eps)}) \ED(\mu) \leq (1+2\eps) \ED(\mu).\]
This completes the proof.
\end{proof}
Low-skew mapping will play an impartial role in our algorithm design and analysis.

\subsection{Reduction to estimating window costs.}\label{sub:reduction}


Let $\cE:\cA\times \cB \to \set{0,\dots, d}$ be an estimate of the edit distance such that  $\cE(a,b) \geq \ED(a,b)$ for all $a\in \cA,b \in \cB$. 
Given a monotone mapping $\mu:\cA\to\cB\cup \set{\perp}$,
we define its cost with respect to $\cE$ as follows:

$$ \ED_{\cE}(\mu)=\sum_{a\in \cA}\cE(a,\mu(a)) + \sum_{i=1}^n \Big|\#\{a \text{~s.t.~}i \in \mu(a)\} -1\Big|$$

We define $\ED(\cE)$ as the minimal cost $\ED_{\cE}(\mu)$ over all monotone mappings. 

Now, given such an estimation the next lemma (combining ideas from~\cite{Ukkonen85} and ~\cite{BEGHS18}) asserts that one can efficiently compute $\ED(\cE)$. 

Instead of computing $\ED(\cE)$ directly, we pick a threshold $\Delta$, and verify whether $\ED(\cE) \leq \Delta$. Indeed, if we can answer whether $\ED(\cE) \leq \Delta$, or $\ED(\cE) > (7+o(1))\Delta$ efficiently, then by increasing the threshold by an $(1+\epsilon)$ factor each time, we will be able to compute $\ED(\cE)$ within a $(7+o(1))$ approximation in $\log_{1+\eps}{n}$ iterations.

\begin{lemma}[Reduction to estimating window costs]\label{lem:DP} 
Given  an estimate $\cE$,  one  can  compute $\ED(\cE)$ in $O(\frac{1}{\eps} \frac{n^2}{d^2} \log{n})$-time.
\end{lemma}

\begin{proof}
Pick a threshold $\Delta$. Set $\gamma=\frac{\Delta d}{n}$. If $\gamma_\tau \leq \gamma$, then from $\cB^{\tau}$ pick every $\lfloor \frac{\gamma}{\gamma_\tau}\rfloor$ windows so that gap between two consecutive windows in $\cB^{\tau}$ is $\geq \gamma-1$ for all $\tau$. Therefore, the total number of windows that we consider in $B$ is $O(\frac{n}{\gamma})$.

For window $w$, we let $e(w)$ denote the index of the last character of $w$; for a set $\cW$, $e(\cW)$ denotes the set of last indices. For $i \in e(\cW)$, we let $\iprev := \max\{i' \in e(\cW) \cup \{0\} \wedge i' < i\}$ denote the index of the previous finish in $\cW$ (or $0$ if such an index does not exist).

We abuse notation and let $\ED(i,j)$ denote the minimum cost of alignment ending at $i,j$ using estimates $\cE$.
Following~\cite{Ukkonen85} we use dynamic programming to fill a table of $\ED(i,j)$ for every pair $i,j \in e(\cA) \times e(\cB)$ such that $|i-j| \leq 10\Delta$.
Notice that the number of such pairs is bounded by:
\begin{gather}\label{eq:pairs} 
|\cA||\cB|\frac{10\Delta}{n}
= O\left( \frac{n}{d} \cdot \frac{n}{\gamma} \cdot \frac{\Delta}{n} \right) 
= O\left( \frac{n^2}{d^2}\right).
\end{gather}

For boundary conditions, we define $\ED(i,j) = \infty$ whenever $|i-j| > 10\Delta$, and $\ED(i,0) = \ED(0,i) = i$.

Consider a pair $i,j$ such that $e(a) = i$ (notice that there may be $O_{\eps}(1)$ $B$-windows ending at $j$).
The cost $\ED(i,j)$ of alignment ending at $i,j$ is given by taking the minimum of: 
\begin{itemize}
\item Cost of deleting the last $A$-window: $i-\iprev + \ED(\iprev,j)$; 
\item Cost of deleting the last several $B$-characters: $j-\jprev + \ED(\iprev,j)$; and
\item Cost of using a last pair of windows: $\min_{e(b) = j}\Big\{\cE(a,b) + \ED(s(a)-1,s(b)-1)\Big\}$.
\end{itemize}

Notice that the runtime of our algorithm is dominated by the number of pairs~\ref{eq:pairs}, aka it is $O( \frac{n^2}{d^2})$. Now, considering all choices of $\Delta$, we get the required running time.


\end{proof}

\subsection{Close-window graphs and the preprocessing phase}
\label{subsec:algo}


For 
subset of windows $\cC$ we define the 
graph $G_{\cC,\tau}$ as follows: The vertex set equals $\cC$. 
The pair $(c,c')$
is connected by an edge if $\ED(c,{c'})\le \tau d$. 
For a substring $z\in \Sigma^d$, we denote by $\cN^{\cC, \tau}(z)$ 
the set of all windows $c\in \cC$ satisfying $\ED(c,z)\le \tau d$. The windows in $\cN^{\cC, \tau}(z)$ will also be referred to as the $\tau$-neighbors of $c$ in $\cC$, and $|\cN^{\cC, \tau}(z)|$ as its degree in $\cC$. When it is clear from the context, we will often omit $\tau$, and simply use the terms such as neighbors and degree of $c$. We will abuse notation and also use this definition for $z \notin \cC$. 

%

\paragraph{Preprocessing Phase algorithm.} 

The preprocessing algorithm crucially uses the algorithm for computing {\em small edit distance} with preprocessing from Section~\ref{sec:small-preprocess}. In the preprocessing phase, strings $A$ and $B$ are processed separately.
Let $\tau \in \{0,\frac{(1+\epsilon)^i}{d}\}$ for $i=[0..\log_{(1+\eps)} d]$. 

\paragraph{Preprocessing $A$.} 
For each $\tau$, the preprocessing algorithm computes the graph $G_{\cA, \tau}$. The number of windows of $A$ is $\frac{n}{d}$. Hence, the preprocessing time over all $\tau$ is $O(\frac{n^2}{d^2}d^2)=O(n^2)$.

\paragraph{Preprocessing $B$.}
The preprocessing algorithm computes $G_{\cB^{\tau}, \tau}$ for each $\tau$. By the preprocessing algorithm of Section~\ref{sec:small-preprocess}, we can process entire $B$ in $O(n\log{n})$ time so that the computation of $\tau d$-thresholded edit distance between any pair of windows can be run in $O(\tau^2d^2\log{n})$ time.

Note that for a given $\tau$, the gap between two consecutive windows in $B$ is $\gamma_\tau=\max{(1, \lfloor{\epsilon \tau d}\rfloor)}$. Therefore, when $\tau=0$, the number of windows is $O(n^2)$, but for every pair of windows, edit distance computation time is $O(1)$. For $\tau > 0$, the number of windows is $O(\frac{n}{\eps \tau d})$. Hence, the computation time is $O(\frac{n^2}{\eps^2 \tau^2d^2} \tau^2d^2\log{n})=O(\frac{n^2}{\eps^2})$. Thus, over all $\tau$, the total preprocessing time for $B$ is $\tilde{O}_{\eps}(n^2)$.



\subsection{Query Phase algorithm} 
The input to the query phase algorithm is the two strings, 
as well as the close-window graphs computed in the preprocessing phase.
The output is an estimate data structure that can answer  $\cE:\cA \times\cB \to \R^+$ queries in $O(\log(n))$ time. 
We implicitly initialize $\cE(a,b)\leftarrow \infty$ 
for all pairs $(a,b)$.



We consider all choices of $\tau \in \{0,\frac{(1+\epsilon)^i}{d}\}$ for $i=[0..\log_{(1+\eps)} d]$. 
For $\tau$ we run the following algorithm that attempts to discover the pairs of windows $a,b \in \cA \times \cB^{\tau}$ of edit distance at most $\tau$.
The estimate algorithm has some false negatives, and it may also have false positives whose true edit distance is up to $7\tau$.
The estimate can be fed into the DP in Section~\ref{sub:reduction}.

Recall that $\ttau:= |\cB^{\tau}|$ denote the number of windows in $\cB^{\tau}$, and $\ttau=\frac{n}{\eps\tau d}$.


For each value of $\tau$, the algorithm below uses $O(\ttau^{4/3+o(1)})$ queries to edit distance of pairs of windows of length $O(d)$ of the form is $\ED(a,b)<\tau d$. Using the algorithm from Theorem~\ref{thm:small-edit}, each query can be answered in $\tilde{O}(d^2 \tau^2)$ time. Hence the total run time is given by 
\begin{gather*}
O\Big(\ttau^{4/3+o(1)} \cdot d^2 \tau^2\Big) = O(n^{3/2+o(1)}).
\end{gather*}


%


\paragraph{Initialization: Covered windows}
Initially, all windows are uncovered.
Intuitively, we say that a window is covered when we have upper bounded the edit distance to its relevant neighbors in $\cB^{\tau}$.

\subsubsection*{A: Intervals}
Consider a partition of $[n]$ into $\ttau^{1/3+\eps}$ contiguous intervals of length $n/\ttau^{1/3+\eps} \leq \ttau^{2/3-\eps} \cdot d$. For $A$ we define the {\em $\cA$-interval} $I_{\cA}$ corresponding to interval $I \subset [n]$ as the set of $\leq \ttau^{2/3-\eps}$ windows with indices in $I$. Therefore, for $A$-windows they are either entirely contained in the interval or don't intersect it.
For $B$, we let $I/\eps$ denote a $1/\eps$-factor expansion of $I$ (i.e. the interval of length $|I|/\eps$ centered at $I$)%
\footnote{For example, if $I = [20,30]$ then its $3$-expansion is $[10,40]$.}. We define the $\cB^{\tau}$-interval $I_{\cB^{\tau}}$ to be the set of  windows that intersect $I/\eps$. When clear from context we sometimes just call $I_{\cA},I_{\cB^{\tau}}$ intervals.
	
\subsubsection*{B: Sampling seeds}	
For each $\cA$-interval $I_{\cA}$, if less than $\log^2(n)$ windows in $I_{\cA}$ remain uncovered, we simply find all of their $\tau$-neighbors in $\cB^{\tau}$ using $\tilde{O}(t_{\tau})$ queries and mark them covered.
Otherwise we sample $\log^2(n)$ uncovered windows 
from $I_{\cA}$. 
For each sampled window $a$, we test whether $\abs {\cN^{\cB^\tau,\tau}(a)}\gtrsim \ttau^{1/3}$. 
This is done by sampling $\ttau^{2/3} \log^2 (n)$ windows $b\in \cB^{\tau}$ and querying $\ED(a,b)$ for each. (We account for those queries later, depending on whether $a$ is dense or sparse.)

If more than $\log^2 (n)/2$ of the samples belong to $\cN^{\cB,\tau}(a)$  then $a$ is declared dense, otherwise it is sparse. 
If the window is dense we process it as described below, after which it is a covered window and no longer a good sample.
We then continue to sample (in random order) other uncovered windows from the same $I_{\cA}$ until: If the number of sparse windows sampled so far is smaller than $\log^2(n)$, we stop sampling whenever we see $\log^2(n)$ consecutive covered windows. Otherwise (the number of observed sparse sampled windows is at least $\log^2(n)$), we stop sampling after querying $\log^2(n)$ consecutive windows which are all either sparse or covered.

\remark{Therefore, if we keep discovering dense windows, we process it as in Step C, the window gets covered, and we keep on sampling more uncovered windows.}

\subsubsection*{C: Dense windows}
Suppose that $a$ is dense; choose (arbitrarily) $b\in \cN^{\cB^{\tau}, \tau}(a)$ among those discovered during Step B while processing $a$. 
	For each windows-pair $a'\in \cN^{\cA, 2\tau}(a)$ and $b'\in \cN^{\cB^{\tau}, 4\tau}(b)$, 
	whose estimate has not been computed yet, the algorithm sets $\cE(a',b')\leftarrow 7\tau$. 
	This is done abstractly by pointing each $a'$ to $a$, $b'$ to $b$ and marking that $\ED(a,b)\le \tau$.
	Observe that the sets $\cN^{\cA,2\tau}(a), \cN^{\cB^{\tau}, 4\tau}(b)$ have already been computed during the preprocessing phase. 
	We mark each window in the set $\cN^{\cA, 2\tau}(a)$ as covered.

\paragraph{Approximation:}
Observe that every $a' \in \cN^{\cA, 2\tau}(a)$ is indeed covered in the sense that by triangle inequality, for every $b'' \in \cN^{\cB, \tau}(a')$
\begin{gather}\label{eq:triangle-4}
\ED(b,b'') \leq \ED(b'',a') + \ED(a',a) + \ED(a,b) \leq 4\tau,
\end{gather}
and hence $\cN^{\cB, \tau}(a') \subseteq \cN^{\cB, 4\tau}(b)$.
Similarly, by triangle inequality for every $b' \in \cN^{\cB, 4\tau}(b)$ and $a' \in \cN^{\cB, 2\tau}(a)$, we have that
\begin{gather}\label{eq:triangle-7}
\ED(a',b') \leq \ED(a',a) + \ED(a,b) + \ED(b,b') \leq 7\tau.
\end{gather}

\paragraph{Complexity:}	
Notice that if the $(\cB^{\tau},\tau)$-neighborhoods of two dense windows (or one dense and one sparse) $a$ and $a'$ intersect, then when we process one of them as dense we will cover both.
Hence we only need to run the dense subroutine at most $\ttau^{2/3}$ times. 
Each run requires $\tilde{O}(\ttau^{2/3})$ queries, and hence in total over the entire $\tau$-th iteration we only need $\tilde{O}(\ttau^{4/3})$ queries.

\subsubsection*{D: Sparse windows}
For each interval $I_\cA$, out of the set of windows $a\in I_\cA$ which were declared sparse, we pick at random a set $S(I_\cA)$ of size $\log^2(n)$. 
For every window in $S(I_\cA)$, we query its entire $(\cB^{\tau},\tau)$-neighborhood using $\ttau$ queries. 
For each interval $I_{\cA}$ we record the union of all $\tilde{O}(\ttau^{1/3})$ intervals $\widehat{I_{\cB^{\tau}}}$ that contain any $(\cB^{\tau},\tau)$-neighbors of any of the sparse samples $a\in S(I_\cA)$. We call these $\cB$-windows the {\em relevant windows} for the windows in $I_{\cA}$.
We henceforth no longer look to match windows from $I_{\cA}$  to irrelevant $\cB$-windows. 
Note that in a low-skew mapping (for more precise statement, see Lemma~\ref{lem:approx}),  windows in $I_{\cA}$ cannot be mapped to any irrelevant $\cB$-windows under that mapping.
(Hence in total across all $\ttau^{1/3+\eps}$ intervals the sparse samples take $\tilde{O}(\ttau^{4/3+\eps})$ queries.)

\paragraph{Approximation:}
Recall that by Lemma~\ref{lem:windowComptabile} and Lemma~\ref{lemma:low-skew}, there is a low-skew monotone mapping that approximates the optimal transformation to within $(1+\eps)$-factor. 
For any low-skew monotone mapping $\mu$, the entire interval $I_{\cA}$  is mapped to a single $\cB^{\tau}$-interval $I_{\cB^{\tau}}$.
Suppose that $(1-\eps)$-fraction of the sparse windows in $I_{\cA}$ are mapped to $\cB^{\tau}$-windows (or $\perp$) of distance greater than $\tau$.
Then we can safely discard the $\tau$-edges for the remaining $\eps$-fraction of sparse windows with negligible loss in approximation factor. 
Hence in total we pay only $(1+O(\eps))$-factor in approximation for sparse windows.
Otherwise, w.h.p. at least one of the samples has a $(\cB^{\tau},\tau)$-neighbor in $I_{\cB^{\tau}}$. 
For more details, see Lemma~\ref{lem:approx}.

\paragraph{Complexity:}
Each uncovered $\cA$-window has only $\tilde{O}\Big(\ttau^{1/3} \cdot \ttau^{2/3-\eps}\Big)= \tilde{O}(\ttau^{1-\eps})$ relevant windows.

\subsubsection*{Recursion}
We recurse on Parts A-D of the algorithm, with the following modifications for the $\ell$-th level of the recursion.
\begin{itemize}
\item We increase the number of intervals to   $\ttau^{1/3+(\ell+1)\eps}$, and their size decreases accordingly to $\tilde{O}(\ttau^{2/3-(\ell+1)\eps})$.
\item We only sample relevant windows when we estimate degrees. The degree-threshold for a window to be considered ``dense'' remains $\ttau^{1/3}$. Notice that a window may be dense with respect to the entire graph, but sparse with respect to its relevant windows. 

\item Once we discover a dense window, we run Part C without regard to relevant/irrelevant windows. In particular the calculation of total number of queries spent on dense windows is global for the entire $\tau$-th iteration of the algorithm, including recursion.
\item For each sparse sample, we only compute the restriction of its $(\cB^{\tau},\tau)$-neighborhood to relevant windows. Hence we only spend $\tilde{O}\Big(\ttau^{1-\ell\eps}\Big)$ queries for each sample, or a total of $\tilde{O}\Big(\ttau^{4/3+\eps}\Big)$ queries across all intervals.
\item The relevant windows for the next level of recursion are a (strict) subset of the relevant windows in the current level.
\end{itemize}

The recursion continues until each interval has less than $\log^2(n)$ windows, 
after which all windows are covered.

We now summarize the approximation factor and the complexity of the query algorithm.

\begin{lemma}[Time Complexity]\label{lem:timequery}
	Let $A,B\in \Sigma^n$. Then running time of the query phase is bounded by: $\tilde{O}_{\eps}(n^{3/2+\eps})$.
\end{lemma}
\begin{proof}
Fix  $\tau$.  By the complexity analysis of Step C, the total number of queries required to cover dense windows over all the recursion steps is $\tilde{O}(\ttau^{4/3})$. 

On the $\ell$-th level of recursion, the number of intervals is $\ttau^{1/3+(\ell+1)\eps}$. For every interval, we pick at most $\log^2{n}$ sparse windows, and query all relevant windows. The number of relevant windows is $\tilde{O}(\ttau^{1-\ell\eps})$. Therefore, on the $\ell$-th level of recursion, the number of queries spent on sparse windows is $\ttau^{4/3+\eps}$. Since the number of levels of recursion is at most $\frac{1}{\eps}+1$, the total number of queries spent on sparse windows is $O(\frac{\ttau^{4/3+\eps}}{\eps} )$. 

Computing $\tau d$-thresholded edit distance between pairs of windows requires time $\tilde{O}(\tau^2d^2)$ (using our algorithm from Section~\ref{sec:small-preprocess}). Therefore, the time complexity for a given $\tau$ over all dense and sparse windows is $\tilde{O}(\frac{\ttau^{4/3+\eps}\tau^2d^2}{\eps} )=\tilde{O}(\frac{n^{3/2+\eps}}{\eps^{7/4}})$. 

Since, the number of choices of $\tau$ is $O(\frac{\log{n}}{\eps})$ and the time to run the DP from Section~\ref{sub:reduction} is $O(\frac{1}{\eps}n^{3/2}\log{n})$ , the overall total time complexity is $\tilde{O}(\frac{n^{3/2+\eps}}{\eps^{15/4}})$. 
\end{proof}

\begin{lemma}[Approximation]\label{lem:approx}
Let $A,B\in \Sigma^n$. Let $\cE:\cA\times \cB \to \R^{+}$ be the cost function produced during the query phase. 
	Then with probability at least $1-\frac{1}{n}$, we have:
	\[ \ED(\cE)\le (7+\eps)\ED(A,B).\]
\end{lemma}
\begin{proof}
Note that if we can construct $\cE:\cA\times \cB \to \set{0,\dots, d}$ such that $\ED(a, b) \leq \cE(a, b)\leq 7\ED(a, b)$, then using the DP algorithm from Section~\ref{sub:reduction} and employing Lemma~\ref{lem:windowComptabile}, we get a $7+o(1)$ approximation for $\ED(A,B)$. Moreover, by Lemma~\ref{lemma:low-skew}, it is only required to compute all the edges of $\cE$ with the above accuracy which an optimum low-skew monotone mapping $\mu$ would use. Fix such a mapping $\mu$.

We prove that for the first level of the recursion, for each interval $I\in I_A$ it is either the case that there exists a sparse window $a$ such that: 
$\mu(a)\in \cN^{\cBt,\tau}(a)$, or that the covered dense windows provide a good approximation for the edges used by $\mu$. 
Indeed, fix $I\in I_A$, the proof proceeds by case analysis.

\textbf{Case 1:} Suppose that there exists $\tau$ such that at least $\eps$-fraction of $a\in I_\cA$, are such that $\mu(a)\in \cN^{\cBt,\tau}(a)$ and $a$ is $\tau$-sparse. 

In this case, with high probability the algorithm will eventually pick a window $a\in I_\cA$ such that $\mu(a)\in \cN^{\cBt,\tau}(a)$ and $a$ is $\tau$-sparse.
Consider the set $\widehat{I_{\cB^{\tau}}}$ recorded by the algorithm. Since $\mu$ is a low-skew mapping, one of the intervals $I_\cB\in \widehat{I_{\cB^{\tau}}}$  is such that all the edges $(a,\mu(a))$ where $a\in I_\cA$, are such that $\mu(a)\in I_\cB$, and hence declared relevant. Therefore, in further iterations of the algorithm these edges will be assigned with the required approximation guarantee.

\textbf{Case 2:} Suppose that for all $\tau$ at most $\eps$-fraction of $a\in I_\cA$, are such that $\mu(a)\in \cN^{\cBt,\tau}(a)$ and $a$ is $\tau$-sparse.

In this case we may fail to detect all the edges $(a,\mu(a))$, where $\mu(a)\in \cN^{\cBt,\tau}(a)$ and $a$ is $\tau$-sparse. Nevertheless, in that case, even if we map {\em all } these edges  to $\perp$, we only lose a $(1+\eps)$ factor in the edit distance. As for the rest of the windows $a\in I_\cA$, we claim that with high probability for at least $1-\eps$ of the windows $a$ we have: $\cE(a,\mu(a))\le 7\ED(a,\mu(a))$.

Indeed, observe that whenever the algorithm completes step B, then it is the case that with high probability all but at most $\eps$-fraction of dense windows are already covered. If this is the case, then for each covered window $a\in I_\cA$ we have: $\cE(a,\mu(a))\le 7\ED(a,\mu(a))$. For the rest we have no guarantee on $\cE(a,\mu(a))$. However, even if we map {\em all } these edges  to $\perp$, we only lose a $(1+\eps)$ factor in the edit distance. The claim follows.
\end{proof}
We therefore have the following theorem.

\begin{theorem}
\label{thm:large-edit}
Given two strings $A, B \in \Sigma^n$, we can approximate $\ED(A, B)$ within $7+o(1)$ approximation with probability at least $1-\frac{1}{n}$ with a preprocessing time of $\tilde{O_{\eps}}(n^2)$ and query time of $\tilde{O_{\eps}}(n^{3/2+o(1)})$.
\end{theorem}

\section{No preprocessing: $3+o(1)$-approx in $n^{1.6+o(1)}$ time}
In this section we introduce our $(3+o(1))$-approximation algorithm for edit distance that runs in time $n^{1.6+o(1)}$ without preprocessing.
At a high level, it is similar to other recent traingle-inequality based approximation algorithms for edit distance. In particular, the previous state of the art algorithm by Andoni~\cite{Andoni19} obtains a similar result when the edit distance is large (near-linear), but we can give an overall faster algorithm using the sublinear algorithm for small edit distance  with preprocessing (Section~\ref{sec:small}). The preprocessing cost is negligble when we apply it once to each window, and use the sublinear algorithm to compute the distances of many pairs.

\begin{theorem}[Approximate edit distance without preprocessing]\label{thm:3+eps} \hfill

Given two strings $A, B \in \Sigma^n$, we can approximate $\ED(A, B)$ within $3+o(1)$ approximation in $\tilde{O_{\eps}}(n^{1.6+o(1)})$ time with probability at least $1-\frac{1}{n}$.
\end{theorem}

\paragraph{High level idea} The algorithm enumerates over various thresholds $\tau$. For each value of $\tau$, the algorithm first marks all the $\cA$-windows as  $\tau$-uncovered. Then, it uses sampling to estimate the degree of each $\cA$ window, and classifies them as sparse or dense. It handles sparse windows similarly to Section~\ref{sec:approx}. As for the dense windows, if there are few of them, it exhaustively finds their $(\cB,\tau)$-neighbors. Otherwise, it sparsifies the set of uncovered dense windows as follows. It enumerates over the set of $\cB$ windows: For each such a window $b$ it estimates its degree with respect to uncovered dense windows. If the degree is large, then it computes $\cN^{\cA,2\tau}(b), \cN^{\cB,\tau}(b)$, marks that the relative distance between pairs in $\cN^{\cA,2\tau}(b)\times \cN^{\cB,\tau}(b)$ as upper bounded by $3\tau$. It then moves each uncovered dense window in $\cN^{\cA,2\tau}(b)$ to the set of covered windows. In such a way, since we remove the neighborhood of dense $\cB$-windows, we show that the number of uncovered dense $\cA$-windows decreases significantly. We recurse on the sparsification phase; each iteration uses a smaller degree threshold for dense $\cB$-windows and handles fewer remaining uncovered dense $\cA$-windows.

\paragraph{Similarities to Section~\ref{sec:approx} Algorithm}
Similar to Section~\ref{sec:approx} and other recent approximation algorithms, we partition the input strings into windows, and consider the close-window graph where two windows share an edge if they are close in edit distance. We handle high-degree (``dense'') windows using triangle inequality, and low-degree (``sparse'') by iteratively focusing on narrowing intervals. 

\paragraph{Main technical difference compared to Section~\ref{sec:approx} Algorithm}
A subtle technicality of this algorithm is that in the sparsification phase, we can remove $\cB$-windows of high degree, and all their $\cA$-neighbors. This suffices to ensure that the remaining $\cA$-windows are sparse {\em on average}. However, the analysis of the sparse case, crucially relies on {\em every window} being sparse. By Markov's inequality, once we decrease the average degree of the $\cA$-window at most $t^{-\eps}$-fraction of them remain overly-dense. We can thus recurse on all the $\cB$-windows and the $t^{-\eps}$-fraction overly-dense $\cA$-windows, again removing the highest-degree $\cB$-windows. After $O(1/\eps)$ iterations, all the dense $\cA$-windows have been removed.\\  

As in Section~\ref{sec:approx}, we repeat the following steps for every $\tau$ in a multiplicative-$(1+\eps)$-net.

\subsubsection*{Parameters and notation}
Following the notation of Section~\ref{sub:windows}, we set the base window length to $d = n^{0.2}$, and the number of $\cA$-windows is $t = n^{0.8}$; the number of windows in $\cBt$ is $\ttau = O_{\eps}(t/\tau)$.
Our algorithm will use $\ttau^{3/2+o(1)}$ queries, each in time $\tilde{O}(d^2 \tau^2)$, as well as the DP from Lemma~\ref{lem:DP}. Hence the total running time is given by 
\begin{gather}\label{eq:3-time}
\tilde{O}\Big(\ttau^{3/2+o(1)} \cdot d^2 \tau^2 + \frac{n^2}{d^2}\Big) = \tilde{O}(n^{1.6+o(1)}).
\end{gather}

Our sparsification phase (Steps A-2 and B below) works in iterations, where in each iteration we cover the edges of the form $(a,b)$ where $b$ is a high degree vertex.
In more detail, the algorithm iteratively identifies $\cB$-windows with high degree. At the first iteration, the degree threshold is $\deg_1 := \ttau^{1/2}$, and it decreases by $\ttau^{\eps}$ in each subsequent iteration. I.e. at the  $g$-th iteration it is $\deg_g := \ttau^{1/2-(g-1)\eps}$.

We maintain a partition of $\cA$ into three subsets: $\cA = \cAs \cup \cAb \cup \cAc$. 
Initially, $|\cAb|\leq |\cA| = t \leq \ttau$.
In each iteration of the sparsification phase, windows from $\cAb$ are moved to $\cAc$.  The upper bound on $|\cAb|$ decreases by $\ttau^{\eps}$-factor in each iteration.



\subsubsection*{Step A: Estimating density of $\cA$-windows}
For each $a \in \cA$, we sample $\ttau^{1/2-\eps}\log^2(n)$ $\cBt$-windows  $b$ at random and query $\ED(b,a)$.
We place $a$ in $\cAb$ if at least $\frac{1}{2}\log^2(n)$ of the samples are within edit distance $\tau$. Otherwise, we place it in $\cAs$ and ignore it until Step C of the algorithm.

\paragraph{Complexity:}
We spend $\tilde{O}(\ttau^{1/2})$ queries for each $a \in \cAb$, hence a total of $\tilde{O}(\ttau^{3/2})$. 

\subsubsection*{Step B-$g$. An iteration of the sparsification phase}
In each iteration of the sparsification phase, we enumerate over the $\cBt$-windows. For each window $b$ that has not already been marked {\em dense} in previous iterations, we sample $\frac{|\cAb| \log^2(n)}{\deg_g}$ $\cAb$-windows  $a$ at random and query $\ED(b,a)$.
We say that $b$ is {\em dense} if at least $\frac{1}{2}\log^2(n)$ of the samples are within edit distance $\tau$.

If $b$ is dense, we query its entire $\cN^{\cAb,\tau}(b),\cN^{\cBt,2\tau}(b)$ neighborhoods.
We (implicitly) add edges with cost $3\tau$ for every pair in $\cN^{\cAb,\tau}(b) \times \cN^{\cBt,2\tau}(b)$, 
and move the windows in  $\cN^{\cAb,\tau}(b)$ to $\cAc$.

If the number of $\cAb$ windows becomes at most $\ttau^{1/2}$ at any point, we exhaustively find all their neighbors in $\cBt$ and move them to $\cAc$.

\paragraph{Approximation}
By triangle inequality, every pair of windows in $\cN^{\cAb,\tau}(b) \times \cN^{\cBt,2\tau}(b)$ has edit distance at most $3\tau$.
Notice also that by triangle inequality $\cN^{\cBt,\tau}(\cN^{\cAb,\tau}(b)) \subseteq \cN^{\cBt,2\tau}(b)$, i.e. we have discovered all the $(\cBt,\tau)$-neighbors of all the $\cAc$-windows. 

\paragraph{Complexity:}
We maintain the bound that at the beginning of the $g$-th iteration, $|\cAb| = O(\ttau^{1-(g-1)}) = O(\ttau^{1/2} \deg_g)$.
Hence, similarly to Step A, we spend $\tilde{O}(\ttau^{1/2})$ queries for estimating the degree of each $b\in \cBt$, for a total of $\tilde{O}(\ttau^{3/2})$. 

Every time we discover a dense $b$, we query its edit distance to $\leq t+\ttau$ windows, and decrease by $\Omega(\deg_{g})$ the number of remaining $\cAb$-windows. Recall that we start the $g$-th iteration with at most $O(\ttau^{1-(g-1)\eps}) =O(\ttau^{1/2}\deg_{g})$ $\cAb$-windows.
Hence in total this step requires $O((t+\ttau)\cdot  \ttau^{1/2})=O(\ttau^{3/2})$ queries.

\subsubsection*{The sparsification phase: iterating over Step B-$g$}
We iteratively apply Step B-$g$ $O(1/\eps)$ times.
At the end of the $g$-th iteration, every remaining $\cBt$-window has at most $\deg_g = \ttau^{1/2-(g-1)\eps}$ remaining $(\cAb,\tau)$-neighbors.
Hence the total number of $\tau$-close pairs in $\cBt \times \cAb$ is $\ttau^{3/2-(g-1)\eps}$. 
Since every $\cAb$ window has $\Omega(\ttau^{1/2+\eps})$ $(\cBt,\tau)$-neighbors%
\footnote{Notice that the number of remaining neighbors for $a \in \cAb$ does not change during the run of the sparsification phase, since once any of $a$'s neighbors is declared {\em dense}, we move $a$ to $\cAc$.}, we have that $|\cAb| = O(\ttau^{1-(g)\eps})$.

\subsubsection*{Step C. Sparse windows}



	

We process the $\cAs$-windows as in the sparse case in Section~\ref{sec:approx} (for completeness, we spell out the details below).
This algorithm is somewhat simpler than Section~\ref{sec:approx} since we already determined in advance which windows are sparse and which are dense.

\paragraph{Intervals (first iteration):}
Consider a partition of $[n]$ into $\ttau^{1/2+2\eps}$ contiguous intervals of length $n/\ttau^{1/2+2\eps} \leq \ttau^{1/2-2\eps} \cdot d$. For $A$ we define the {\em $\cA$-interval} $I_{\cA}$ corresponding to interval $I \subset [n]$ as the set of $\leq \ttau^{1/2-2\eps}$ windows with indices in $I$. Therefore, for $A$-windows they are either entirely contained in the interval or don't intersect it.
For $B$, we let $I/\eps$ denote a $1/\eps$-factor expansion of $I$ (i.e. the interval of length $|I|/\eps$ centered at $I$)%
\footnote{For example, if $I = [20,30]$ then its $3$-expansion is $[10,40]$.}. We define the $\cB^{\tau}$-interval $I_{\cB^{\tau}}$ to be the set of windows that intersect $I/\eps$. When clear from context we sometimes just call $I_{\cA},I_{\cB^{\tau}}$ intervals.

\paragraph{Sparse subroutine (first iteration):}
For each interval $I_\cA$, if at most $\log^2(n)$ of its windows are sparse, we simply query their entire $(\cB^{\tau},\tau)$-neighborhoods.
Otherwise, we sample a random set $S(I_\cA)$ of $\log^2(n)$ windows from $I_\cA \cap \cAs$. 
For every window in $S(I_\cA)$, we query its entire $(\cB^{\tau},\tau)$-neighborhood using $\ttau$ queries. 
For each interval $I_{\cA}$ we record the union of all $\tilde{O}(\ttau^{1/2+\eps})$ intervals $\widehat{I_{\cB^{\tau}}}$ that contain any $(\cB^{\tau},\tau)$-neighbors of any of the sparse samples $a\in S(I_\cA)$. We call these $\cB$-windows the {\em relevant windows} for the windows in $I_{\cA}$.
We henceforth no longer look to match windows from $I_{\cA}$  to irrelevant $\cB$-windows. 
Note that in a low-skew mapping, if at least one of the samples is matched, then windows in $I_{\cA}$ cannot be mapped to any irrelevant $\cB$-windows under that mapping.

\paragraph{Approximation (first iteration):}
Recall that by Lemma~\ref{lem:windowComptabile} and Lemma~\ref{lemma:low-skew}, there is a low-skew monotone mapping that approximates the optimal transformation to within $(1+O(\eps))$-factor. 
For any low-skew monotone mapping $\mu$, the entire interval $I_{\cA}$  is mapped to a single $\cB^{\tau}$-interval $I_{\cB^{\tau}}$.
Suppose that $(1-\eps)$-fraction of the sparse windows in $I_{\cA}$ are mapped to $\cB^{\tau}$-windows (or $\perp$) of distance greater than $\tau$.
Then we can safely discard the $\tau$-edges for the remaining $\eps$-fraction of sparse windows with negligible loss in approximation factor. 
Hence in total we pay only $(1+O(\eps))$-factor in approximation for sparse windows.
Otherwise, w.h.p. at least one of the samples has a $(\cB^{\tau},\tau)$-neighbor in $I_{\cB^{\tau}}$. 

\paragraph{Complexity (first iteration):}
Each sparse $\cA$-window has only $\tilde{O}\Big(\ttau^{1/2+\eps} \cdot \ttau^{1/2-2\eps}\Big)= \tilde{O}(\ttau^{1-\eps})$ relevant windows.
Since there are $\ttau^{1/2+2\eps}$ $\cA$-intervals, we spend use a total of $\tilde{O}(\ttau^{3/2+\eps})$ queries.


\paragraph{Recursion}
We recurse on the sparse subroutine, with the following modifications for the $\ell$-th iteration.
\begin{itemize}
\item We increase the number of intervals to   $\ttau^{1/2+(\ell+2)\eps}$, and their size decreases accordingly to $\tilde{O}(\ttau^{1/2-(\ell+2)\eps})$.
\item For each sparse sample, we only compute the restriction of its $(\cB^{\tau},\tau)$-neighborhood to relevant windows. Hence we only spend $\tilde{O}\Big(\ttau^{1-\ell\eps}\Big)$ queries for each sample, or a total of $\tilde{O}\Big(\ttau^{3/2+\eps}\Big)$ queries across all intervals.
\item The relevant windows for the next level of recursion are a (strict) subset of the relevant windows in the current level.
\end{itemize}
The recursion continues until each interval has less than $\log^2(n)$ sparse windows, 
after which we can simply query the distance of every remaining sparse window to all its relevant $\cBt$-windows.

\subsubsection*{Completing the proof of Theorem~\ref{thm:3+eps}}
As we argued above, the algorithm finds a $(3+\eps)$-approximation using $\tilde{O}(\ttau^{3/2+\eps})$ queries. 
Taking $\eps$ to be slightly sub-constant completes the proof of Theorem~\ref{thm:3+eps}. \qed

\section{Hardness}\label{sec:hardness}

In this section we formalize, in the context of (approximate) edit distance, the folklore intuition (based on~\cite{WW18}) that polynomial preprocessing can not circumvent fine-grained complexity lower bounds.
In Subsection~\ref{sub:SETH} we show that known fine-grained complexity hardness results for exact edit distance and related problems extend to accommodate polynomial preprocessing.

In Subsection~\ref{sub:Hardness-of-Approx} we consider the problem of edit distance approximation. There are essentially no conditional hardness results for this problem, and in fact recent work obtained a truly-subquadratic constant factor approximation algorithm~\cite{CDGKS18}. Improving this factor, and in particular obtaining a truly-subquadratic $(1+\eps)$-approximation factor, is perhaps the most important open problem in this area. There are evidences that providing $1+o(1)$-factor approximation might be hard, as it implies new circuit lower bounds~\cite{AB17}.
Theorem~\ref{thm:Hardness-of-Approx} shows that essentially any approximation factor that is obtained with polynomial preprocessing can also be obtained without it. 
Note that this holds unconditionally, even if (BP)-SETH is false.

\subsection{SETH-hardness of exact string alignment with preprocessing}\label{sub:SETH}
The Strong Exponential Time Hypothesis is an (extreme) strengthening  of $\P\neq \NP$ postulating that $k$-SAT on $n$ variables requires $2^{(1-\delta_k)n}$ time. Building on~\cite{AHWW16}, we can prove our hardness based on the milder BP-SETH which replaces $k$-CNF with a branching program:

\begin{hypothesis}[BP-SETH]
Given a branching program over $n$ variables of width $W$ and length $T$ such that $T^W = 2^{o(n)}$, deciding whether it has a satisfying assignment requires time $2^{(1-o(1))n}$  time.
\end{hypothesis}

\begin{theorem}[BP-SETH hardness]\label{thm:SETH} \hfill

Unless BP-SETH is false, there is no algorithm that preprocesses two input strings in polynomial time and then computes their (edit distance / longest common subsequence / dynamic time warping) in truly-subquadratic time.
\end{theorem}

\begin{remark}
We remark that unlike with $k$-SAT,  it is plausible that the brute-force algorithm for BP-SAT is optimal to within $\poly(n)$ factors, and in fact better algorithms would imply new circuit lower bounds (\cite{AHWW16} and references therein). Under a corresponding strengthening of BP-SETH one can show that string alignment with preprocessing requires $N^2/\polylog(N)$ time.
\end{remark}

The proof of Theorem~\ref{thm:SETH} builds on {\em alignment gadgets} and {\em normalized vector gadgets (NVG)} from previous works on SETH and BP-SETH hardness of string alignment~\cite{BI18,BK15,AHWW16}. Each NVG represents a half-assignment to the branching program, and the alignment gadgets define a composition of the NVGs into two long strings.
Here we deviate from typical SETH-hardness proofs of sequence similarity, and use a divide-and-conquer approach of~\cite{WW18} to construct two larger sets of shorter strings. This allows us to reuse the preprocessing of each shorter string when we compare every pair to look for a satisfying assignment (aka a satisfying pair of half-assignments).

Below we use $\dist()$ to refer to the distance under the relevant similarity measure (edit distance / longest common subsequence / dynamic time warping); for longest common subsequence we use the ``co-LCS'' (edit-distance-without-substitutions) distance $\dist(A,B) := n-\LCS(A,B)$.

\subsubsection*{Normalized Vector Gadgets}
Given a BP $\varphi$ of width $W$ and length $T$, {\em normalized vector gadgets ($\NVG$)} map half assignments $a,b \in \{0,1\}^{n/2}$ into strings such that:
\begin{gather*}
\dist(\NVG_A(a),\NVG_B(b)) = 
\begin{cases}
	c_T & \text{if assignment $(a \circ b)$ satisfies $\varphi$};\\
	c_F & \text{otherwise}
\end{cases},
\end{gather*}
where $c_T < c_F$ are integers that depend on $W,T$.

\begin{lemma}[Normalized Vector Gadgets~\cite{AHWW16}]\label{lem:NVG}\hfill

Given a BP of width $W$ and length $T$, we can construct NVGs of length $T^{O(\log(W)}))$ for all half assignments $a,b \in \{0,1\}^{n/2}$ in time $2^{n/2} \cdot T^{O(\log(W)}))$.
\end{lemma}

\subsubsection*{Alignment Gadgets}
Consider two ordered sets of strings $\cA, \cB$ of cardinalities $n_A < n_B$, respectively.
An alignment $\mu$ is a monotone partial mapping from $\cA$ to $\cB \cup \{\perp\}$.
An alignment $\mu$ is {\em structured} if it maps the $i$-th string in $\cA$ to the $i+\Delta$ string in $\cB$ for some fixed shift $\Delta$ and for all $a \in \cA$.

The $\acost$ of a mapping $\mu$ is defined by:
\begin{gather*}
\acost(\mu,\cA,\cB) := \sum_{a \in \cA} \dist(a,\mu(a)).
\end{gather*}
Here $\dist(a,\perp):= \max_{a'\in\cA,b\in\cB}\dist(a',b)$.


An {\em alignment gadget} is a mapping from $\cA,\cB$ into respective strings $\AG_A(\cA),\AG_B(\cB)$ such that for some parameter $c_{\GA}=c(n_A,n_B)$:
\begin{gather}\label{eq:GA}
\min_{\text{alignment $\mu$}}\acost(\mu,\cA,\cB) \leq \dist(\AG_A(\cA),\AG_B(\cB)) +c_{\GA}\leq \min_{\text{structured alignment $\mu$}} \acost(\mu,\cA,\cB).
\end{gather}

\begin{lemma}[Alignment gadgets~\cite{BK15}]\label{lem:GA} \hfill

Edit distance, LCS (with binary alphabet), and Dynamic Time Warping admit alignment gadgets that can be computed in linear time. 
\end{lemma}

\subsubsection*{Completing the proof of BP-SETH-hardness}

\begin{proof}[Proof of Theorem~\ref{thm:SETH}]
Suppose that we have an algorithm that computes $\dist()$ for strings of length $N$ with preprocessing time $O(N^{\pretime})$ and query time $O(N^{2-\eps})$.
Given a BP over $n = 2\log_2(N)$ variables, we construct all its normalized vector gadgets in near-linear time as in Lemma~\ref{lem:NVG}. 

We partition the $\cA$-NVGs into $2^{(1-1/\pretime)n}$ subsets $\cA_1, \dots \cA_{2^{(1-1/\pretime)n}}$ of size $2^{n/\pretime}$ each (and likewise for $\cB$).
For each subset $\cA_i$, we construct its alignment gadget $A_i$ of size $\tilde{O}(N^{1/\pretime})$. 
For $\cB_j$, let $B_j$ be constructed by the alignment gadget for the set repeated twice.
If no pair of half-assignments corresponding to $\cA_i \times \cB_j$ satisfies the BP, then every pair of NVGs is at distance $c_F$, and by~\eqref{eq:GA} the distance of $A_i,B_j$ will be $d_F := 2^{n/\pretime}c_F - c_{\GA}$.
If there is a satisfying pair, then the structured alignment that matches the corresponding NVGs will have cost at most
\begin{gather*}
d_T := c_T + (2^{n/\pretime}-1)c_F - c_{\GA} < d_F.
\end{gather*}

We preprocess all the strings in total time $\tilde{O}\big(N^{1-1/\pretime}\cdot (N^{1/\pretime})^{\pretime}\big) = \tilde{O}(N^{2-1/\pretime})$.
Finally, we compute the distance between all $(N^{1-1/\pretime})^2$ pairs in time $O\big(N^{2-2/\pretime} \cdot (N^{1/\pretime})^{2-\eps}\big) = O\big(N^{2-2\eps/\pretime}\big)$.
The BP is satisfiable iff at least one of the pairs is at distance at most $d_T$.
\end{proof}

\subsection{Preprocessing doesn't help for approximate $\ED$ in truly-subquadratic time} \label{sub:Hardness-of-Approx}

\begin{theorem}[Hardness of Approximation]\label{thm:Hardness-of-Approx}\hfill

If there is an $\alpha$-approximation algorithm for edit distance that runs in polynomial preprocessing time and truly-subquadratic query time, then there is an $(\alpha+o(1))$-approximation algorithm that runs in truly-subquadratic time with no preprocessing.
\end{theorem}

The proof combines the divide-and-conquer steps from our approximate edit distance algorithm (Section~\ref{sec:approx}) with that of~\cite{WW18} (see also last step in the proof of Theorem~\ref{thm:SETH}).

\begin{proof}
Suppose that there exists an algorithm that computes an $\alpha$-approximation of edit distance using $O(n^{\pretime})$-preprocessing and $O(n^{2-\eps})$-query time. 
First, we assume wlog that the true edit distance is  $k = \omega(n^{1-1/2\pretime})$, otherwise we can solve the problem in time $O(n^{2-1/\pretime})$ using the algorithm of~\cite{LMS98}. In particular, we can henceforth neglect additive errors of $O(n^{1-1/2\pretime})$.

Using the notation of Section~\ref{sub:windows}, we decompose the strings into windows with base width $d := n^{1/\pretime}$. The $A$-windows have no overlap, and for the $B$-windows we consider $\cB^{\tau}$ for $\tau = n^{-1/2\pretime}$.
Hence we have
$\tilde{O}(n^{1-1/\pretime})$ $\cA$-windows and $\tilde{O}(n^{1-1/2\pretime})$ $\cB$-windows, all of length $O(n^{1/\pretime})$.

We preprocess all the windows in time $\tilde{O}\big(n^{1-1/2\pretime} \cdot (n^{1/\pretime})^{\pretime}\big) = \tilde{O}\big(n^{2-1/2\pretime}\big)$.

We then run the $\alpha$-approximate edit distance algorithm on pairs of windows. By the argument of~\cite{Ukkonen85}, it suffices to only compute the distances between pairs of windows whose starting points are within $\pm k$ far apart.
In particular for every $\cA$-window, we only need to compute the edit distance to $\tilde{O}(n^{1-1/\pretime})$ $\cB$-windows.
In total we spend $\tilde{O}\big((n^{1-1/\pretime})^2 \cdot (n^{1/\pretime})^{2-\eps}\big) = \tilde{O}\big(n^{2-\eps/\pretime}\big)$ time on this phase.

Given the $\alpha$-approximate window-window distance estimates, we aggregate them in time $\tilde{O}(n^{2-2/\pretime})$ using Lemma~\ref{lem:DP}.
Thus we obtain an $\alpha$-approximation to the optimal window-compatible matching, which by Lemma~\ref{lem:windowComptabile}, is an $(\alpha+o(1))$-approximation to the edit distance.
\end{proof}

\bibliographystyle{alpha}
\bibliography{references}

\end{document}